\documentclass[uplatex, a4paper, 10pt, english]{article}
\usepackage{bxpapersize}

\usepackage{amsthm}
\newtheorem{thm}{Theorem}
\newtheorem{dfn}[thm]{Definition}
\newtheorem{lem}[thm]{Lemma}

\usepackage{amsmath}
\usepackage{amssymb}

\usepackage{authblk}

\usepackage{algorithm}
\usepackage{algorithmicx}
\usepackage{algpseudocode}

\makeatletter
\renewcommand{\fnum@algorithm}{\fname@algorithm}
\newenvironment{procedure}[1][htb]{
\renewcommand{\ALG@name}{Procedure} \begin{algorithm}[#1]}{\end{algorithm}}

\makeatother

\usepackage{caption}
\usepackage[disable,colorinlistoftodos,prependcaption]{todonotes}
\usepackage{color}

\usepackage[hidelinks]{hyperref}
\usepackage{ifthen}

\usepackage[all]{xy}
\CompileMatrices

\usepackage{tikz}
\usetikzlibrary{
    positioning,
    calc,
    decorations.pathmorphing
}

\newcounter{nodecount}

\makeatletter
    \DeclareRobustCommand{\rvdots}{%
        \vbox{
            \baselineskip4\p@\lineskiplimit\z@
            \kern-\p@
            \hbox{.}\hbox{.}\hbox{.}
        }
    }
\makeatother

\tikzset{
    graphnode/.style={
        draw,
        thick,
        circle,
        minimum height=\nodesize+1pt,
        anchor=center,
        inner sep=0em
    },
    Tedge/.style={
        thick,
    },
    Sedge/.style={
        dashed,
    },
    graphedge/.style={
        draw
    },
    pathedge/.style={
        draw,
        decorate,
        decoration={snake, 
            pre length=3pt, 
            post length=3pt, 
            segment length=6pt, 
            amplitude=1pt
        },
    },
    dots/.style args={#1per #2}{%
        line cap=round,
        line width=0.75pt,
        dash pattern=on 0 off #2/#1,
    },
}

\def\nodesize{5mm}
\def\unitspacing{\nodesize*2}
\usepackage{ifthen}
\newcommand{\bentpath}[3][(0, 0)]{
    \path[pathedge] #2 .. controls ($0.5*#2+0.5*#3+#1$) .. #3;
}
\newcommand{\partline}[5][dot]{
    \coordinate (center) at ($(\unitspacing*1.8*1.5*#2,\unitspacing*#3) - (\unitspacing*1.8*0.75,0)$);
    \ifthenelse{\equal{#1}{dot}}{
        \draw[dots=5 per 1cm] ($(center) + (0, \unitspacing*#4+\nodesize*0.5)$) -- ($(center) + (0, \unitspacing*#5-\nodesize*0.5)$);
    }{}
    \ifthenelse{\equal{#1}{thick}}{
        \draw[thick] ($(center) + (0, \unitspacing*#4+\nodesize*0.5)$) -- ($(center) + (0, \unitspacing*#5-\nodesize*0.5)$);
    }{}
}

\newcommand{\partlineh}[5][dot]{
    \coordinate (center) at ($(\unitspacing*1.8*1.5*#2,\unitspacing*#3) - (\unitspacing*1.8*0.75,0)$);
    \ifthenelse{\equal{#1}{dot}}{
        \draw[dots=5 per 1cm] ($(center) + (\unitspacing*1.8*1.5*#4-\nodesize*0.5, 0)$) -- ($(center) + (\unitspacing*1.8*1.5*#5+\nodesize*0.5, 0)$);
    }{}
    \ifthenelse{\equal{#1}{thick}}{
        \draw[thick] ($(center) + (\unitspacing*1.8*1.5*#4-\nodesize*0.5, 0)$) -- ($(center) + (\unitspacing*1.8*1.5*#5+\nodesize*0.5, 0)$);
    }{}
}

\newcommand{\partcaption}[3]{
    \coordinate (center) at ($(\unitspacing*1.8*1.5*#1,\unitspacing*#2)$);
    \node at (center) () {#3};
}

\newcommand{\setspacingcoordinate}{
    \coordinate (hs) at (\horizontalspacing,0);
    \coordinate (vs) at (0, \verticalspacing);
    \coordinate (bn) at (0,\bentspacing);
    \coordinate (be) at (\bentspacing,0);
    \coordinate (bw) at (-1*\bentspacing,0);
    \coordinate (bs) at (0,-1*\bentspacing);
    \coordinate (mh) at (1.5*\nodesize,0);
    \coordinate (mv) at (0,0.5*\verticalspacing);
    \coordinate (l) at ($-0.5*(hs)$);
    \coordinate (r) at ($0.5*(hs)$);
}

\newcommand{\EXEVEN}[5]{
    \def\set{#3}
    \def\sub{#4}
    \def\horizontalspacing{\unitspacing*2*0.9}
    \def\verticalspacing{\unitspacing}
    \def\bentspacing{\unitspacing*0.3}
    \setspacingcoordinate
    \coordinate (center) at ($#1*1.5*(hs)+#2*(vs)$);
    \node[graphnode] at ($(center)+(r)+1.5*(vs)$) (p1) {};
    \node[graphnode] at ($(center)+(r)+0.5*(vs)$) (p2) {};
    \node[graphnode] at ($(center)+(l)+0.5*(vs)$) (p3) {};
    \node[graphnode] at ($(center)+(l)+1.5*(vs)$) (p4) {};
    \node[graphnode] at ($(center)+(l)-0.5*(vs)$) (x1) {};
    \node[graphnode] at ($(center)+(l)-2.5*(vs)$) (x2) {};
    \node[graphnode] at ($(center)+(r)-0.5*(vs)$) (y1) {};
    \node[graphnode] at ($(center)+(r)-1.5*(vs)$) (y2) {};
    \node[graphnode] at ($(center)+(l)-1.5*(vs)$) (z1) {};
    \node at ($(center)-3.2*(vs)$) () {#5};

    \node[graphnode] at ($(center)+(r)-2.5*(vs)$) (w1) {};

    \ifthenelse{\equal{\set}{S} \OR \equal{\set}{A}}{
        \ifthenelse{\equal{\set}{S}}{
            \ifthenelse{\sub=0 \OR \sub=1 \OR \sub=2 \OR \sub=3 \OR \sub=4}{
                \foreach \u/\v in {p1/p4,p2/p3,x1/z1,y1/z1,x2/w1,y2/w1}{
                    \path[draw] (\u) edge[graphedge] (\v);
                }
                \ifthenelse{\sub=0 \OR \sub=1 \OR \sub=3}{
                    \path[draw] (p3) edge[graphedge] (p4);
                }{}
                \ifthenelse{\sub=0 \OR \sub=2 \OR \sub=4}{
                    \path[draw] (p1) edge[graphedge] (p2);
                }{}
                \ifthenelse{\sub=0 \OR \sub=3 \OR \sub=4}{
                    \path[draw] (x1) edge[graphedge] (y1);
                    \path[draw] (x2) edge[graphedge] (y2);
                }{}
                \ifthenelse{\sub=0}{
                    \path[draw] (x2) edge[graphedge] node[above] {$e_a$} (y2);
                    \path[draw] (x1) edge[graphedge] node[above] {$e_b$} (y1);
                    \path[draw] (p1) edge[graphedge] node[left] {$e_c$} (p2);
                    \path[draw] (p3) edge[graphedge] node[right] {$e_d$} node[left] {$C$} (p4);

                }{}
            }{}
        }{}

        \ifthenelse{\equal{\set}{A}}{
            \ifthenelse{\sub=0 \OR \sub=1 \OR \sub=2}{
                \ifthenelse{\sub=0 \OR \sub=1}{
                    \path[draw] (p2) edge[graphedge] (x1);
                    \path[draw] (y1) edge[graphedge] (x2);
                    \draw (y2) edge[out=0, in=0, distance=\unitspacing] (p1);
                }{}
                \ifthenelse{\sub=0 \OR \sub=2}{
                    \path[draw] (p3) edge[graphedge] (y1);
                    \path[draw] (x1) edge[graphedge] (y2);
                    \draw (x2) edge[out=180, in=180, distance=\unitspacing] (p4);
                }{}
            }{}
        }{}

        \ifthenelse{\equal{\set}{A} \OR \sub=1 \OR \sub=2}{
            \node[graphnode] at (x1) {$x_1$};
            \node[graphnode] at (x2) {$x_2$};
            \node[graphnode] at (y1) {$y_1$};
            \node[graphnode] at (y2) {$y_2$};
            \ifthenelse{\(\equal{\set}{A} \AND \sub=0\) \OR \sub=1}{
                \node[graphnode] at (p1) {$p_1$};
                \node[graphnode] at (p2) {$p_2$};
            }{}
            \ifthenelse{\(\equal{\set}{A} \AND \sub=0\) \OR\sub=2}{
                \node[graphnode] at (p3) {$p_3$};
                \node[graphnode] at (p4) {$p_4$};
            }{}
        }{}
    }{}

    \ifthenelse{\equal{\set}{T} \OR \equal{\set}{B}}{
        \ifthenelse{\equal{\set}{T}}{
            \ifthenelse{\sub=0 \OR \sub=1 \OR \sub=2 \OR \sub=3 \OR \sub=4}{
                \foreach \u/\v in {p1/p4,p2/y1,p3/x1,x2/w1,y2/z1}{
                    \path[draw] (\u) edge[graphedge] (\v);
                }
                \ifthenelse{\sub=1 \OR \sub=3}{
                    \path[draw] (p1) edge[graphedge] (p2);
                }{}
                \ifthenelse{\sub=2 \OR \sub=4}{
                    \path[draw] (p3) edge[graphedge] (p4);
                }{}
                \ifthenelse{\sub=1 \OR \sub=2}{
                    \path[draw] (x1) edge[graphedge] (y1);
                    \path[draw] (x2) edge[graphedge] (y2);
                }{}
            }{}
        }{}

        \ifthenelse{\equal{\set}{B}}{
            \ifthenelse{\sub=0 \OR \sub=1}{
                \draw (z1) edge[out=180, in=180, distance=\unitspacing] (p3);
                \path[draw] (p4) edge[graphedge] (w1);
                
            }{}
            \ifthenelse{\sub=0 \OR \sub=2}{
                \draw (w1) edge[out=0, in=0, distance=\unitspacing] (p2);
                \path[draw] (p1) edge[graphedge] (z1);
            }{}
        }{}

        \ifthenelse{\equal{\set}{B} \OR \sub=1 \OR \sub=2}{
            \node[graphnode] at (z1) {$z_1$};
            \node[graphnode] at (w1) {$w_1$};
            \ifthenelse{\(\equal{\set}{B}\AND \sub=0\) \OR \sub=1}{
                \node[graphnode] at (p3) {$p_3$};
                \node[graphnode] at (p4) {$p_4$};
            }{}
            \ifthenelse{\(\equal{\set}{B}\AND \sub=0\) \OR \sub=2}{
                \node[graphnode] at (p1) {$p_1$};
                \node[graphnode] at (p2) {$p_2$};
            }{}
        }{}
    }{}
    
}

\newcommand{\EXODD}[5]{
    \def\set{#3}
    \def\sub{#4}
    \def\horizontalspacing{\unitspacing*2*0.9}
    \def\verticalspacing{\unitspacing}
    \def\bentspacing{\unitspacing*0.3}
    \setspacingcoordinate
    \coordinate (center) at ($#1*1.5*(hs)+#2*(vs)$);
    \node[graphnode] at ($(center)+(l)+0.5*(vs)$) (v0) {};
    \node[graphnode] at ($(center)+(l)+1.5*(vs)$) (v1) {};
    \node[graphnode] at ($(center)+(l)+2.5*(vs)$) (v2) {};
    \node[graphnode] at ($(center)+(r)+2.5*(vs)$) (v3) {};
    \node[graphnode] at ($(center)+(r)+1.5*(vs)$) (v4) {};
    \node[graphnode] at ($(center)+(r)+0.5*(vs)$) (v5) {};
    \node[graphnode] at ($(center)$) (c) {};
    \node[graphnode] at ($(center)+(l)-0.5*(vs)$) (l1) {};
    \node[graphnode] at ($(center)+(l)-1.5*(vs)$) (l2) {};
    \node[graphnode] at ($(center)+(l)-2.5*(vs)$) (l3) {};
    \node[graphnode] at ($(center)+(l)-3.5*(vs)$) (l4) {};

    \node[graphnode] at ($(center)+(r)-0.5*(vs)$) (r1) {};
    \node[graphnode] at ($(center)+(r)-1.5*(vs)$) (r2) {};
    \node[graphnode] at ($(center)+(r)-2.5*(vs)$) (r3) {};
    \node[graphnode] at ($(center)+(r)-3.5*(vs)$) (r4) {};

    \ifthenelse{\sub=0 \OR \sub=1 \OR \sub=2}{

        \node[graphnode] at (v5) (q4) {};
        \node[graphnode] at (c) (q0) {};
        \node[graphnode] at (r3) (x1) {};
        \node[graphnode] at (r2) (x2) {};
        \node[graphnode] at (l3) (y1) {};
        \node[graphnode] at (r1) (y2) {};

        \node[graphnode] at (l4) (z1) {};
        \node[graphnode] at (l2) (z2) {};
        \node[graphnode] at (r4) (w1) {};
        \node[graphnode] at (l1) (w2) {};
    }{}

    \ifthenelse{\sub=3 \OR \sub=4 \OR \sub=5}{
        \node[graphnode] at (r3) (x1) {};
        \node[graphnode] at (c) (x2) {};
        \node[graphnode] at (r4) (y1) {};
        \node[graphnode] at (l1) (y2) {};
        \node[graphnode] at (v0) (z1) {};
        \node[graphnode] at (l2) (z2) {};
        \node[graphnode] at (l3) (z3) {};
        \node[graphnode] at (r1) (w1) {};
        \node[graphnode] at (r2) (w2) {};
        \node[graphnode] at (l4) (w3) {};
    }{}
    \node at ($(center)-4.2*(vs)$) () {#5};

    \ifthenelse{\equal{\set}{S} \OR \equal{\set}{A}}{
        \ifthenelse{\equal{\set}{S}}{
            \ifthenelse{\sub=0 \OR \sub=1 \OR \sub=2 \OR \sub=3 \OR \sub=4 \OR \sub=5}{
                \foreach \u/\v in {v0/v5,v2/v3,c/r1,l1/l2,l2/r2,l3/l4,l4/r4}{
                    \path[draw] (\u) edge[graphedge] (\v);
                }
                \ifthenelse{\sub=0 \OR \sub=1 \OR \sub=2 \OR \sub=3}{
                    \path[draw] (c) edge[graphedge] (l1);
                    \path[draw] (r3) edge[graphedge] (r4);
                }{}
                \ifthenelse{\sub=0 \OR \sub=1 \OR \sub=3 \OR \sub=4 \OR \sub=5}{
                    \path[draw] (v1) edge[graphedge] (v0);
                }{}
                \ifthenelse{\sub=0 \OR \sub=2 \OR \sub=3 \OR \sub=4 \OR \sub=5}{
                    \path[draw] (v4) edge[graphedge] (v3);
                }{}
                \ifthenelse{\sub=0 \OR \sub=3 \OR \sub=4 \OR \sub=5}{
                    \path[draw] (l3) edge[graphedge] (r3);
                    \path[draw] (r1) edge[graphedge] (r2);
                }{}
                \ifthenelse{\sub=0 \OR \sub=1 \OR \sub=2 \OR \sub=3 \OR \sub=4}{
                    \path[draw] (v4) edge[graphedge] (v5);
                }{}
                \ifthenelse{\sub=0 \OR \sub=1 \OR \sub=2 \OR \sub=3 \OR \sub=5}{
                    \path[draw] (v1) edge[graphedge] (v2);
                }{}

                \ifthenelse{\sub=0 \OR \sub=3}{
                    \node[graphnode] at (v0) {$v_0$};
                    \node[graphnode] at (v1) {$v_1$};
                    \node[graphnode] at (v2) {$v_2$};
                    \node[graphnode] at (v3) {$v_3$};
                    \node[graphnode] at (v4) {$v_4$};
                    \node[graphnode] at (v5) {$v_5$};
                    \node[graphnode] at (r3) {$q$};
                    \ifthenelse{\sub=0}{
                        \path[draw] (l3) edge[graphedge] node[above] {$f$} (r3);
                        \path[draw] (r3) edge[graphedge] node[left] {$f^\prime$} (r4);
                        \path[draw] (r1) edge[graphedge] node[left] {$e$} (r2);
                        \path[draw] (c) edge[graphedge] node[above left] {$e^\prime$} (l1);  
                        \node at ($(v2)-0.25*(hs)-0.5*(vs)$) () {$C^{*}$};
                        \node at ($(l3)-0.25*(hs)-0.5*(vs)$) () {$C^{**}$};
                    }{}
                    
                }{} 

            }{}
        }{}

        \ifthenelse{\equal{\set}{A}}{
            \ifthenelse{\sub=0 \OR \sub=1 \OR \sub=2}{
                \ifthenelse{\sub=0 \OR \sub=1}{
                    \draw (x1) edge[out=0, in=0, distance=\unitspacing] (v3);
                    \path[draw] (y1) edge[graphedge] (x2);
                    \draw (y2) edge[out=0, in=0, distance=0.5*\unitspacing] (v4);
                    \draw (y1) edge[out=180, in=180, distance=\unitspacing,color=white] (v1);
                }{}
                \ifthenelse{\sub=0 \OR \sub=2}{
                    \draw (y1) edge[out=180, in=180, distance=\unitspacing] (v1);
                    \draw (x1) edge[out=0, in=0, distance=0.5*\unitspacing] (y2);

                    \draw plot[smooth,tension=0.] coordinates{(x2.135) ($0.5*(q0)+0.5*(w2)$) (v0.-80)};
                    
                }{}
            }{}
        }{}

        \ifthenelse{\equal{\set}{A}}{
            \ifthenelse{\sub=3 \OR \sub=4 \OR \sub=5}{
                \ifthenelse{\sub=3 \OR \sub=4}{
                    \draw (v2) edge[out=180, in=180, distance=\unitspacing] ($(z3.180)+0.5*(vs)$);
                    \draw ($(z3.180)+0.5*(vs)$) edge ($(x1)+0.5*(vs)-0.5*(hs)$);
                    \draw ($(x1)+0.5*(vs)-0.5*(hs)$) edge (x1);
                    \draw (y1) edge[out=0, in=0, distance=0.5*\unitspacing] ($(x2)+(r)+0.5*(\nodesize,0)$);
                    \draw ($(x2)+(r)+0.5*(\nodesize,0)$) edge (x2);

                    \draw (y2) edge[out=180, in=180, distance=0.5*\unitspacing] (v1);
                }{}
                \ifthenelse{\sub=3 \OR \sub=5}{
                    \draw (y1) edge[out=0, in=0, distance=\unitspacing] (v4);
                    \path[draw] (x1) edge[graphedge] (y2);
                    \path[draw] (x2) edge[graphedge] (v5);
                    
                }{}
            }{}
        }{}

        \ifthenelse{\(\equal{\set}{A} \AND \sub=0\) \OR \sub=1 \OR \sub=2}{
            \node[graphnode] at (x1) {$x_1$};
            \node[graphnode] at (x2) {$x_2$};
            \node[graphnode] at (y1) {$y_1$};
            \node[graphnode] at (y2) {$y_2$};
            \ifthenelse{\(\equal{\set}{A} \AND \sub=0\) \OR \sub=1}{
                \node[graphnode] at (v4) {$v_4$};
                \node[graphnode] at (v3) {$v_3$};
            }{}
            \ifthenelse{\(\equal{\set}{A} \AND \sub=0\) \OR \sub=2}{
                \node[graphnode] at (v1) {$v_1$};
                \node[graphnode] at (v0) {$v_0$};
            }{}
        }{}

        \ifthenelse{\(\equal{\set}{A} \AND \sub=3\) \OR \sub=4 \OR \sub=5}{
            \node[graphnode] at (x1) {$x^\prime_1$};
            \node[graphnode] at (x2) {$x^\prime_2$};
            \node[graphnode] at (y1) {$y^\prime_1$};
            \node[graphnode] at (y2) {$y^\prime_2$};
            \ifthenelse{\(\equal{\set}{A} \AND \sub=3\) \OR \sub=4}{
                \node[graphnode] at (v1) {$v_1$};
                \node[graphnode] at (v2) {$v_2$};
            }{}
            \ifthenelse{\(\equal{\set}{A} \AND \sub=3\) \OR \sub=5}{
                \node[graphnode] at (v4) {$v_4$};
                \node[graphnode] at (v5) {$v_5$};
            }{}
        }{}
    }{}

    \ifthenelse{\equal{\set}{T} \OR \equal{\set}{B}}{
        \ifthenelse{\equal{\set}{T}}{
            \ifthenelse{\sub=0 \OR \sub=1 \OR \sub=2}{
                \foreach \u/\v in {v0/c,v1/v2,v2/v3,v4/v5,r1/l1,r2/l2,l3/l4,r3/r4}{
                    \path[draw] (\u) edge[graphedge] (\v);
                }
                \ifthenelse{\sub=1}{
                    \path[draw] (v4) edge[graphedge] (v3);
                }{}
                \ifthenelse{\sub=2}{
                    \path[draw] (v1) edge[graphedge] (v0);
                }{}
                \ifthenelse{\sub=1 \OR \sub=2}{
                    \path[draw] (x1) edge[graphedge] (y1);
                    \path[draw] (x2) edge[graphedge] (y2);
                }{}
            }{}
            \ifthenelse{\sub=3 \OR \sub=4 \OR \sub=5}{
                \foreach \u/\v in {v2/v3,v3/v4,v1/v5,v0/l1,c/r1,l2/r2,l3/r3,l4/r4}{
                    \path[draw] (\u) edge[graphedge] (\v);
                }
                \ifthenelse{\sub=4}{
                    \path[draw] (v1) edge[graphedge] (v2);
                }{}
                \ifthenelse{\sub=5}{
                    \path[draw] (v4) edge[graphedge] (v5);
                }{}
                \ifthenelse{\sub=4 \OR \sub=5}{
                    \path[draw] (x1) edge[graphedge] (y1);
                    \path[draw] (x2) edge[graphedge] (y2);
                }{}
            }{}
            \ifthenelse{\sub=0 \OR \sub=3}{
                \node[graphnode] at (v0) {$v_0$};
                \node[graphnode] at (v1) {$v_1$};
                \node[graphnode] at (v2) {$v_2$};
                \node[graphnode] at (v3) {$v_3$};
                \node[graphnode] at (v4) {$v_4$};
                \node[graphnode] at (v5) {$v_5$};
                \node[graphnode] at (r3) {$q$};
            }{}
            
        }{}

        \ifthenelse{\equal{\set}{B}}{
            \ifthenelse{\sub=0 \OR \sub=1}{
                \draw (z1) edge[out=180, in=180, distance=\unitspacing] (v1);
                \path[draw] (w1) edge[graphedge] (z2);
                \path[draw] (w2) edge[graphedge] (q0);
                \path[draw] (v0) edge[graphedge] (q4);
                
            }{}
            \ifthenelse{\sub=0 \OR \sub=2}{
                \draw (w1) edge[out=0, in=0, distance=\unitspacing] (v3);
                \draw (z1) edge[out=180, in=180, distance=0.5*\unitspacing] (w2);
                \draw plot[smooth,tension=0.1] coordinates{(z2.25) ($0.5*(q0)+0.5*(y2)$) (q4.-100)};
                \path[draw] (v4) edge[graphedge] (q0);
            }{}
        }{}

        \ifthenelse{\equal{\set}{B}}{
            \ifthenelse{\sub=3 \OR \sub=4}{
                \path[draw] (v4) edge[graphedge] (z1);
                \path[draw] (w1) edge[graphedge] (z2);
                \path[draw] (w2) edge[graphedge] (z3);
                \draw plot[smooth,tension=0.1] coordinates{(w3.30) ($0.5*(x2)+0.5*(w1)$) (v5.-160)};
            }{}
            \ifthenelse{\sub=3 \OR \sub=5}{
                \path[draw] (v2) edge[graphedge] (w1);
                \draw plot[smooth,tension=0.1] coordinates{(z1.-80) ($0.5*(x2)+0.5*(y2)$) (w2.150)};
                \draw (w3) edge[out=180, in=180, distance=0.5*\unitspacing] (z2);

                \draw (z3) edge[out=180, in=180, distance=\unitspacing] (v1);
            }{}
        }{}

        \ifthenelse{\(\equal{\set}{B}\AND \sub=0\) \OR \sub=1 \OR \sub=2}{
            \node[graphnode] at (z1) {$z_1$};
            \node[graphnode] at (w1) {$w_1$};
            \node[graphnode] at (z2) {$z_2$};
            \node[graphnode] at (w2) {$w_2$};
            \node[graphnode] at (q0) {$r_0$};
            \node[graphnode] at (q4) {$r_4$};
            \ifthenelse{\(\equal{\set}{B}\AND \sub=0\) \OR \sub=1}{
                \node[graphnode] at (v1) {$v_1$};
                \node[graphnode] at (v0) {$v_0$};
            }{}
            \ifthenelse{\(\equal{\set}{B}\AND \sub=0\) \OR \sub=2}{
                \node[graphnode] at (v4) {$v_4$};
                \node[graphnode] at (v3) {$v_3$};
            }{}
        }{}

        \ifthenelse{\(\equal{\set}{B}\AND \sub=3\) \OR \sub=4 \OR \sub=5}{
            \node[graphnode] at (z1) {$z^\prime_1$};
            \node[graphnode] at (w1) {$w^\prime_1$};
            \node[graphnode] at (z2) {$z^\prime_2$};
            \node[graphnode] at (w2) {$w^\prime_2$};
            \node[graphnode] at (z3) {$z^\prime_3$};
            \node[graphnode] at (w3) {$w^\prime_3$};
            \ifthenelse{\(\equal{\set}{B}\AND \sub=3\) \OR \sub=4}{
                \node[graphnode] at (v4) {$v_4$};
                \node[graphnode] at (v5) {$v_5$};
            }{}
            \ifthenelse{\(\equal{\set}{B}\AND \sub=3\) \OR \sub=5}{
                \node[graphnode] at (v1) {$v_1$};
                \node[graphnode] at (v2) {$v_2$};
            }{}
        }{}
    }{}
    
}

\newcommand{\STAB}[7]{
    \def\set{#3}
    \def\type{#4}
    \def\sub{#5}
    \def\d{#6}
    
    \ifthenelse{\equal{\set}{A} \OR \equal{\set}{S} \OR \equal{\set}{SA}}{
        \ifthenelse{\sub=-1 \OR \sub=0 \OR \sub=1 \OR \sub=2}{
            \def\lnl{x}
            \def\rnl{y}
        }{
            \def\lnl{y}
            \def\rnl{x}
        }
        
        \def\nn{k}
        \def\ucpl{P}
        \def\cpl{Q}
    }{}
    \ifthenelse{\equal{\set}{B} \OR \equal{\set}{T} \OR \equal{\set}{TB}}{
        \ifthenelse{\sub=-1 \OR \sub=0 \OR \sub=1 \OR \sub=2}{
            \def\lnl{w}
            \def\rnl{z}
        }{
            \def\lnl{z}
            \def\rnl{w}
        }
        \def\nn{d}
        \def\ucpl{R}
        \def\cpl{O}
    }{}
    \def\horizontalspacing{\unitspacing*2*0.9}
    \def\verticalspacing{\unitspacing}
    \def\bentspacing{\unitspacing*0.3}
    \setspacingcoordinate
    \coordinate (center) at ($#1*1.5*(hs)+#2*(vs)$);

    \coordinate (pos) at ($(center) + (l) -0.5*(vs)$);
    \foreach \i in {1, ..., \d} {
        \path (pos) node [graphnode] (l\i) {$\lnl_\i$} ++(hs) node [graphnode] (r\i) {$\rnl_\i$};
        \coordinate (pos) at ($(pos)-(vs)$);
    }
    \node at ($(pos)+0.5*(hs)+0.5*(vs)$) (cu) {};
    \node at ($(pos)+0.5*(hs)+0.25*(vs)$) (vdots) {\rvdots};
    
    \coordinate (pos) at ($(pos)-0.5*(vs)$);
    \node at ($(pos)+0.5*(hs)+0.5*(vs)$) (cl) {};
    \path (pos) node [graphnode] (l\nn) {$\lnl_{\nn}$} ++(hs) node [graphnode] (r\nn) {$\rnl_{\nn}$};
    \node at ($(pos)+0.5*(hs)-0.8*(vs)$) () {#7};
    
    \ifthenelse{\(\equal{\set}{S} \OR \equal{\set}{T} \OR \equal{\set}{SA} \OR \equal{\set}{TB}\) \AND \(\NOT \sub=-1\)}{
        \foreach \i in {1, ..., \d,\nn} {
            \bentpath[(bs)]{(l\i)}{(r\i)}
            \path (l\i) edge[draw=none] node {$\cpl_\i$} (r\i);
        }
    }{}
    \ifthenelse{\equal{\set}{S} \OR \equal{\set}{A} \OR \equal{\set}{SA} \OR \sub=3 \OR \sub=4 \OR \sub=5}{
        \coordinate (pl) at ($(center)+0.5*(vs)$);
        \coordinate (pu) at ($(pl)+(vs)$);
    }{
        \coordinate (pl) at ($(center)+2*(vs)$);
        \coordinate (pu) at ($(pl)+(vs)$);
    }
    
    \ifthenelse{\type=1 \OR \type=2}{
        \node[graphnode] at (pl) (p2) {};
        \node[graphnode] at (p2) (p3) {};
    }{
        \node[graphnode] at ($(r)+(pl)$) (p2) {};
        \node[graphnode] at ($(l)+(pl)$) (p3) {};
    }
    
    \ifthenelse{\type=5}{
        \node[graphnode] at (pu) (p1) {};
        \node[graphnode] at (p1) (p4) {};
    }{
        \node[graphnode] at ($(r)+(pu)$) (p1) {};
        \node[graphnode] at ($(l)+(pu)$) (p4) {};
    }

    \ifthenelse{\sub=-1 \OR \sub=0 \OR \sub=1 \OR \sub=2}{
        \node[graphnode] at (p1) (p1) {$p_1$};
        \node[graphnode] at (p2) (p2) {$p_2$};
        \ifthenelse{\NOT \type=1 \AND \NOT\type=2}{
            \node[graphnode] at (p3) (p3) {$p_3$};
        }{}
        \ifthenelse{\NOT \type=5}{
            \node[graphnode] at (p4) (p4) {$p_4$};
        }{}
    }{

        \node[graphnode] at (p2) (p2) {$v_1$};
        \node[graphnode] at (p3) (p3) {$v_3$};
        \node[graphnode] at ($0.5*(p2)+0.5*(p3)$) (v2) {$v_2$};
        \node[graphnode] at (p4) (p4) {$v_4$};
        \ifthenelse{\NOT \type=5}{
            \node[graphnode] at (p1) (p1) {$v_0$};
        }{}

    }
    \ifthenelse{\equal{\set}{T} \OR \equal{\set}{B} \OR \equal{\set}{TB}}{
        \coordinate (mq) at ($(0, 1.5*\verticalspacing)$);
        \coordinate (qu) at ($(pu)+(mq)$);
        \coordinate (q1) at ($(p1)+(mq)$);
        \coordinate (q2) at ($(p2)-(mq)$);
        \coordinate (q3) at ($(p3)-(mq)$);
        \coordinate (q4) at ($(p4)+(mq)$);
        \ifthenelse{\sub=-1 \OR \sub=0 \OR \sub=1 \OR \sub=2}{
            \ifthenelse{\type=1 \OR \type=2}{
                \node[graphnode] at (q2) (q2) {$q_2$};
                \node[graphnode] at (q3) (q3) {};
            }{}
            \ifthenelse{\type=2 \OR \type=4 \OR \type=5}{
                \node[graphnode] at (q1) (q1) {$q_1$};
                \node[graphnode] at (q4) (q4) {\ifthenelse{\type=5}{}{$q_4$}};
            }{}
        }{
            \ifthenelse{\type=4 \OR \type=5}{
                \node[graphnode] at (q1) (q1) {\ifthenelse{\type=5}{}{$r_0$}};
                \node[graphnode] at (q4) (q4) {$r_4$};
            }{} 
        }
    }{}
    
    \ifthenelse{\equal{\set}{S} \OR \equal{\set}{SA}}{ 
        \ifthenelse{\sub=0 \OR \sub=1 \OR \sub=2}{
            \ifthenelse{\type=3 \OR \type=4 \OR \type=5}{
                \bentpath[(bs)]{(p2)}{(p3)}
            }{}
            \ifthenelse{\NOT \type=5}{
                \bentpath[(bn)]{(p4)}{(p1)}
            }{}
            \ifthenelse{\NOT \sub=2}{
                \path[draw] (p3) edge[graphedge] (p4);
                \path (p3) edge[draw=none] node[right] {$\ucpl_1$} (p4);
            }{}
            \ifthenelse{\NOT \sub=1}{
                \path[draw] (p1) edge[graphedge] (p2);
                \path (p2) edge[draw=none] node[left] {$\ucpl_2$} (p1);
            }{}
        }{}
        \ifthenelse{\sub=3 \OR \sub=4 \OR \sub=5}{
            \path[draw] (p2) edge[graphedge] (v2);
            \path[draw] (p3) edge[graphedge] (v2);
            \ifthenelse{\NOT \type=5}{
                \bentpath[(bn)]{(p4)}{(p1)}
            }{}
            \ifthenelse{\sub=5}{
                \path[draw] (p3) edge[graphedge] (p4);
            }{}
            \ifthenelse{\sub=4}{
                \path[draw] (p1) edge[graphedge] (p2);
            }{}
        }{}
    }{
    }
    
    \ifthenelse{\equal{\set}{A} \OR \equal{\set}{SA}}{
        \ifthenelse{\sub=0 \OR \sub=1 \OR \sub=3 \OR \sub=5}{
            \ifthenelse{\sub=0 \OR \sub=1}{
                \path[draw] (p2) edge[graphedge] (l1);
            }{
                \path[draw] (p2.270) edge[graphedge] (l1);
            }
            \foreach \i [remember=\i as \lasti (initially 1)] in {2,...,\d}{
                \path[draw] (r\lasti) edge[graphedge] (l\i);
            }
            \path[draw] (r\d) edge[graphedge] (cu);
            \path[draw] (cl) edge[graphedge] (l\nn);

            \ifthenelse{\type=5}{
                \draw (r\nn) edge[out=0, in=0, distance=\unitspacing] ($(pu)+(r)+0.5*(\nodesize,0)$);
                \draw ($(pu)+(r)+0.5*(\nodesize,0)$) edge (p1);
            }{
                \draw (r\nn) edge[out=0, in=0, distance=\unitspacing] (p1);
            }
        }{}
        \ifthenelse{\sub=0 \OR \sub=2 \OR \sub=3 \OR \sub=4}{
            \ifthenelse{\sub=0 \OR \sub=2}{
                \path[draw] (p3) edge[graphedge] (r1);
            }{
                \path[draw] (p3.270) edge[graphedge] (r1);
            }
            \foreach \i [remember=\i as \lasti (initially 1)] in {2,...,\d}{
                \path[draw] (l\lasti) edge[graphedge] (r\i);
            }
            \path[draw] (l\d) edge[graphedge] (cu);
            \path[draw] (cl) edge[graphedge] (r\nn);

            \ifthenelse{\type=5}{
                \draw (l\nn) edge[out=180, in=180, distance=\unitspacing] ($(pu)+(l)-0.5*(\nodesize,0)$);
                \draw ($(pu)+(l)-0.5*(\nodesize,0)$) edge (p4);
            }{
                \draw (l\nn) edge[out=180, in=180, distance=\unitspacing] (p4);
            }
        }{}
    }{}
    \ifthenelse{\equal{\set}{T} \OR \equal{\set}{TB}}{
        \ifthenelse{\sub=0 \OR \sub=1 \OR \sub=2}{
            \ifthenelse{\type=1 \OR \type=2}{
                \ifthenelse{\sub=0 \OR \sub=1}{
                    \bentpath[(bw)]{(p2)}{(q2)}
                }{
                    \bentpath[(be)]{(p2)}{(q2)}
                }
            }{
                \path (p2) edge[pathedge,out=-90,in=-90] (p3);
            }
            \ifthenelse{\type=2 \OR \type=4 \OR \type=5}{
                \ifthenelse{\NOT \type=5}{
                    \bentpath[(be)]{(p1)}{(q1)}
                    \bentpath[(bw)]{(p4)}{(q4)}
                }{
                    \ifthenelse{\sub=0 \OR \sub=1}{
                        \bentpath[(bw)]{(p1)}{(q1)}
                    }{
                        \bentpath[(be)]{(p1)}{(q1)}
                    }
                }
            }{
                \path (p1) edge[pathedge,out=90,in=90] (p4);
            }
            \ifthenelse{\sub=2}{
                \path[draw] (p3) edge[graphedge] (p4);
            }{}
            \ifthenelse{\sub=1}{
                \path[draw] (p1) edge[graphedge] (p2);
            }{}
        }{}
        \ifthenelse{\sub=3 \OR \sub=4 \OR \sub=5}{
            \path[draw] (p2) edge[graphedge] (v2);
            \path[draw] (p3) edge[graphedge] (v2);
            \ifthenelse{\type=4 \OR \type=5}{
                \ifthenelse{\NOT \type=5}{
                    \bentpath[(be)]{(p1)}{(q1)}
                    \bentpath[(bw)]{(p4)}{(q4)}
                }{
                    \ifthenelse{\sub=3 \OR \sub=4}{
                        \bentpath[(bw)]{(p1)}{(q1)}
                    }{
                        \bentpath[(be)]{(p1)}{(q1)}
                    }
                }
            }{
                \path (p1) edge[pathedge,out=90,in=90] (p4);
            }
            \ifthenelse{\sub=4}{
                \path[draw] (p3) edge[graphedge] (p4);
            }{}
            \ifthenelse{\sub=5}{
                \path[draw] (p1) edge[graphedge] (p2);
            }{}
        }{}
    }{}
    \ifthenelse{\equal{\set}{B} \OR \equal{\set}{TB}}{
        \ifthenelse{\sub=0 \OR \sub=1 \OR \sub=3 \OR \sub=5}{
            \foreach \i [remember=\i as \lasti (initially 1)] in {2,...,\d}{
                \path[draw] (l\lasti) edge[graphedge] (r\i);
            }
            \path[draw] (l\d) edge[graphedge] (cu);
            \path[draw] (cl) edge[graphedge] (r\nn);
            \ifthenelse{\type=1 \OR \type=2}{
                \path[draw] (q3) edge[graphedge] (r1);
            }{
                \path[draw] (p3) edge[graphedge] (r1);
            }
            \ifthenelse{\type=1 \OR \type=3}{
                \draw (ld) edge[out=180, in=180, distance=\unitspacing] (p4);
            }{
                \ifthenelse{\NOT \type=5}{
                    \draw (ld) edge[out=180, in=180, distance=\unitspacing] (q4);
                    \path[draw] (p4) edge[graphedge] (q1);
                }{
                    \draw (l\nn) edge[out=180, in=180, distance=\unitspacing] ($(qu)+(l)-0.5*(\nodesize,0)$);
                    \draw ($(qu)+(l)-0.5*(\nodesize,0)$) edge (q4);
                }
            }
        }{}
        
        \ifthenelse{\sub=0 \OR \sub=2 \OR \sub=3 \OR \sub=4}{
            \foreach \i [remember=\i as \lasti (initially 1)] in {2,...,\d}{
                \path[draw] (r\lasti) edge[graphedge] (l\i);
            }
            \path[draw] (r\d) edge[graphedge] (cu);
            \path[draw] (cl) edge[graphedge] (l\nn);
            
            \ifthenelse{\type=1 \OR \type=2}{
                \path[draw] (q2) edge[graphedge] (l1);
            }{
                \path[draw] (p2) edge[graphedge] (l1);
            }
            \ifthenelse{\type=1 \OR \type=3}{
                \draw (rd) edge[out=0, in=0, distance=\unitspacing] (p1);
            }{
                \ifthenelse{\NOT \type=5}{
                    \draw (rd) edge[out=0, in=0, distance=\unitspacing] (q1);
                    \path[draw] (p1) edge[graphedge] (q4);
                }{
                    \draw (r\nn) edge[out=0, in=0, distance=\unitspacing] ($(qu)+(r)+0.5*(\nodesize,0)$);
                    \draw ($(qu)+(r)+0.5*(\nodesize,0)$) edge (q1);
                }
            }
        }{}
    }{}
}

\newcommand{\CDH}{
    \begin{tikzpicture}
        \def\horizontalspacing{\unitspacing}
        \def\verticalspacing{\unitspacing}
        \def\bentspacing{\unitspacing*0.3}
        \setspacingcoordinate
        \coordinate (center) at ($0*1.5*(hs)+0*(vs)$);
        \node[graphnode] at ($(center) + 0.5*(hs) + 0.5*(vs)$) (p2) {$p_2$};
        \node[graphnode] at ($(center) - 0.5*(hs) + 0.5*(vs)$) (p3) {$p_3$};
        \node[graphnode] at ($(center) + 0.5*(hs) - 0.5*(vs)$) (x1) {$x_1$};
        \node[graphnode] at ($(center) - 0.5*(hs) - 0.5*(vs)$) (y1) {$y_1$};
        \node[graphnode] at ($(center) + 2.5*(hs)$) (p1) {};
        \node[graphnode] at ($(center) - 2.5*(hs)$) (p4) {};
        \path (p2) edge[graphedge] (x1);
        \path (p3) edge[graphedge] (y1);
        \path (p2) edge[pathedge,out=60,in=90] (p1);
        \path (x1) edge[pathedge,out=-60,in=-90] (p1);
        \path (p3) edge[pathedge,out=120,in=90] (p4);
        \path (y1) edge[pathedge,out=240,in=-90] (p4);

        \node at ($(center) - 1.7*(vs)$) () {$C$};

        \coordinate (center) at ($4*(hs)+0*(vs)$);

        \coordinate (center) at ($8*(hs)+0*(vs)$);
        \node[graphnode] at ($(center) + 0.5*(hs) + 0.5*(vs)$) (p2) {$p_2$};
        \node[graphnode] at ($(center) - 0.5*(hs) + 0.5*(vs)$) (p3) {$p_3$};
        \node[graphnode] at ($(center) + 0.5*(hs) - 0.5*(vs)$) (x1) {$x_1$};
        \node[graphnode] at ($(center) - 0.5*(hs) - 0.5*(vs)$) (y1) {$y_1$};
        \node[graphnode] at ($(center) + 2.5*(hs)$) (p1) {};
        \node[graphnode] at ($(center) - 2.5*(hs)$) (p4) {};
        \path (p2) edge[graphedge] (y1);
        \path (p3) edge[graphedge] (x1);
        \path (p2) edge[pathedge,out=60,in=90] (p1);
        \path (x1) edge[pathedge,out=-60,in=-90] (p1);
        \path (p3) edge[pathedge,out=120,in=90] (p4);
        \path (y1) edge[pathedge,out=240,in=-90] (p4);

        \node at ($(center) - 1.7*(vs)$) () {$H$};
    \end{tikzpicture}
}

\newcommand{\testall}[4][\STAB]{
    \def\TEST{#1}
    \begin{figure}[H]
        \begin{tikzpicture}
            \foreach \j/\S/\s in {0/#2/#3,1/#2/#4}{
                \foreach \i/\t in {0/3,1/4,2/5}{
                    \ifthenelse{\NOT \i=0}{
                        \partline{\i}{-11*\j}{4.5}{-4}
                    }{}
                    \TEST{\i}{-11*\j}{\S}{\t}{\s}{2}{}
                }
            }
        \end{tikzpicture}
        \caption{$#2_{#3}$ and $#2_{#4}$}
    \end{figure}
    }

\newcommand{\test}[1][\EXEVEN]{
    \def\TEST{#1}
    \begin{figure}[H]
        \begin{tikzpicture}
            \foreach \j/\S in {0/S,1/T}{
                \foreach \i/\s in {0/0,1/1,2/2,3/3,4/4,5/5}{
                    \ifthenelse{\NOT \i=0}{
                        \partline{\i}{-11*\j}{4.5}{-4}
                    }{}
                    \TEST{\i}{-11*\j}{\S}{\s}{$\S_\s$}
                }
            }
        \end{tikzpicture}
    \end{figure}

    \begin{figure}[H]
        \begin{tikzpicture}
            \foreach \j/\S in {0/A,1/B}{
                \foreach \i/\s in {0/0,1/1,2/2,3/3,4/4,5/5}{
                    \ifthenelse{\NOT \i=0}{
                        \partline{\i}{-11*\j}{4.5}{-4}
                    }{}
                    \TEST{\i}{-11*\j}{\S}{\s}{$\S_\s$}
                }
            }
        \end{tikzpicture}
    \end{figure}
    }

\newcommand{\figexeven}{
    \begin{tikzpicture}
        \EXEVEN{0}{0}{S}{0}{$S$}
        \EXEVEN{2}{0}{T}{0}{$T$}
    \end{tikzpicture}
}

\newcommand{\figexevene}{
    \begin{tikzpicture}
        \EXEVEN{0}{-6}{S}{3}{$S\setminus\{e_c\}$}
        \EXEVEN{1.5}{-6}{T}{3}{$T\cup\{e_c\}$}

        \EXEVEN{3.5}{-6}{S}{4}{$S\setminus\{e_d\}$}
        \EXEVEN{5}{-6}{T}{4}{$T\cup\{e_d\}$}
    \end{tikzpicture}
}

\newcommand{\figexevenst}{
    \begin{tikzpicture}
        \EXEVEN{0}{0}{S}{1}{$S_1$}
        \EXEVEN{1.5}{0}{T}{1}{$T_1$}
        \EXEVEN{3.5}{0}{S}{2}{$S_2$}
        \EXEVEN{5}{0}{T}{2}{$T_2$}
    \end{tikzpicture}       
}

\newcommand{\figexevena}{
    \begin{tikzpicture}
        \EXEVEN{0}{0}{A}{1}{$A_1$}
        \EXEVEN{2}{0}{A}{2}{$A_2$}
        \EXEVEN{4}{0}{A}{0}{$A_1 \cup A_2$}
    \end{tikzpicture}
}

\newcommand{\figexevenb}{
    \begin{tikzpicture}
        \EXEVEN{0}{-6}{B}{1}{$B_1$}
        \EXEVEN{2}{-6}{B}{2}{$B_2$}
        \EXEVEN{4}{-6}{B}{0}{$B_1\cup B_2$}
    \end{tikzpicture}
}

\newcommand{\figevens}{
    \begin{tikzpicture}
        \STAB{0}{0}{S}{1}{1}{2}{$S_1$}
        \STAB{1.5}{0}{S}{1}{2}{2}{$S_2$}
        \partcaption{0.75}{-4}{Case 1: $p_2=p_3$}

        \STAB{3.5}{0}{S}{3}{1}{2}{$S_1$}
        \STAB{5}{0}{S}{3}{2}{2}{$S_2$}
        \partcaption{4.25}{-4}{Case 2: $p_2\not=p_3$}
    \end{tikzpicture}
}

\newcommand{\figevena}{
    \begin{tikzpicture}
        \STAB{0}{0}{A}{1}{1}{2}{$A_1$}
        \STAB{2}{0}{A}{1}{2}{2}{$A_2$}
        \STAB{4}{0}{A}{1}{0}{2}{$A_1 \cup A_2$}
        \partcaption{2}{-4.5}{Case 1: $p_2=p_3$}
        
        \STAB{0}{-7}{A}{3}{1}{2}{$A_1$}
        \STAB{2}{-7}{A}{3}{2}{2}{$A_2$}
        \STAB{4}{-7}{A}{3}{0}{2}{$A_1 \cup A_2$}
        \partcaption{2}{-11.5}{Case 2: $p_2\not=p_3$}
    \end{tikzpicture}
}

\newcommand{\figevent}{
    \begin{tikzpicture}
        \STAB{0}{0}{T}{1}{1}{2}{$T_1$}
        \STAB{1.5}{0}{T}{1}{2}{2}{$T_2$}
        \partcaption{0.75}{-4.5}{Case 1}

        \STAB{3.5}{0}{T}{2}{1}{2}{$T_1$}
        \STAB{5}{0}{T}{2}{2}{2}{$T_2$}
        \partcaption{4.25}{-4.5}{Case 2}

        \STAB{0}{-10}{T}{3}{1}{2}{$T_1$}
        \STAB{1.5}{-10}{T}{3}{2}{2}{$T_2$}
        \partcaption{0.75}{-14.5}{Case 3}

        \STAB{3.5}{-10}{T}{4}{1}{2}{$T_1$}
        \STAB{5}{-10}{T}{4}{2}{2}{$T_2$}
        \partcaption{4.25}{-14.5}{Case 4}
    \end{tikzpicture}
}

\newcommand{\figevenb}[1][1]{
    \begin{tikzpicture}
        \STAB{0}{0}{B}{#1}{1}{2}{$B_1$}
        \STAB{2}{0}{B}{#1}{2}{2}{$B_2$}
        \STAB{4}{0}{B}{#1}{0}{2}{$B_1 \cup B_2$}
        \partcaption{2}{-4.5}{Case #1}
    \end{tikzpicture}
}

\newcommand{\figexodd}{
    \begin{tikzpicture}
        \EXODD{0}{0}{S}{0}{$S$}
        \EXODD{2}{0}{T}{0}{$T$}
        \EXODD{4}{0}{T}{3}{$T^\prime$}{}
    \end{tikzpicture}
}

\newcommand{\figexoddst}{
    \begin{tikzpicture}
        \EXODD{0}{0}{S}{1}{$S_1$}
        \EXODD{1.5}{0}{T}{1}{$T_1$}
        \EXODD{3.5}{0}{S}{2}{$S_2$}
        \EXODD{5}{0}{T}{2}{$T_2$}
    \end{tikzpicture}       
}

\newcommand{\figexoddstp}{
    \begin{tikzpicture}
        \EXODD{0}{0}{S}{4}{$S^\prime_1$}
        \EXODD{1.5}{0}{T}{4}{$T^\prime_1$}
        \EXODD{3.5}{0}{S}{5}{$S^\prime_2$}
        \EXODD{5}{0}{T}{5}{$T^\prime_2$}
    \end{tikzpicture}       
}

\newcommand{\figexodda}{
    \begin{tikzpicture}
        \EXODD{0}{0}{A}{1}{$A_1$}
        \EXODD{2}{0}{A}{2}{$A_2$}
        \EXODD{4}{0}{A}{0}{$A_1 \cup A_2$}
    \end{tikzpicture}
}

\newcommand{\figexoddb}{
    \begin{tikzpicture}
        \EXODD{0}{-7}{B}{1}{$B_1$}
        \EXODD{2}{-7}{B}{2}{$B_2$}
        \EXODD{4}{-7}{B}{0}{$B_1\cup B_2$}
    \end{tikzpicture}
}

\newcommand{\figexoddap}{
    \begin{tikzpicture}
        \EXODD{0}{0}{A}{4}{$A^\prime_1$}
        \EXODD{2}{0}{A}{5}{$A^\prime_2$}
        \EXODD{4}{0}{A}{3}{$A^\prime_1 \cup A^\prime_2$}
    \end{tikzpicture}
}

\newcommand{\figexoddbp}{
    \begin{tikzpicture}
        \EXODD{0}{0}{B}{4}{$B^\prime_1$}
        \EXODD{2}{0}{B}{5}{$B_2$}
        \EXODD{4}{0}{B}{3}{$B^\prime_1\cup B^\prime_2$}
    \end{tikzpicture}
}

\newcommand{\figodds}{
    \begin{tikzpicture}
        \STAB{0}{0}{S}{3}{4}{2}{$S_1$}
        \STAB{1.5}{0}{S}{3}{5}{2}{$S_2$}
        \partcaption{0.75}{-4.5}{Case 1: $|C^*| > 4$}

        \STAB{3.5}{0}{S}{5}{4}{2}{$S_1$}
        \STAB{5}{0}{S}{5}{5}{2}{$S_2$}
        \partcaption{4.25}{-4.5}{Case 2: $|C^*| = 4$}
    \end{tikzpicture}
}

\newcommand{\figodda}{
    \begin{tikzpicture}
        \STAB{0}{0}{A}{3}{4}{2}{$A_1$}
        \STAB{2}{0}{A}{3}{5}{2}{$A_2$}
        \STAB{4}{0}{A}{3}{3}{2}{$A_1 \cup A_2$}
        \partcaption{2}{-4.5}{Case 1: $|C^*| > 4$}

        \STAB{0}{-7}{A}{5}{4}{2}{$A_1$}
        \STAB{2}{-7}{A}{5}{5}{2}{$A_2$}
        \STAB{4}{-7}{A}{5}{3}{2}{$A_1 \cup A_2$}
        \partcaption{2}{-11.5}{Case 2: $|C^*| = 4$}
    \end{tikzpicture}
}

\newcommand{\figoddt}{
    \begin{tikzpicture}
        \STAB{0}{0}{T}{3}{4}{2}{$T_1$}
        \STAB{1.5}{0}{T}{3}{5}{2}{$T_2$}
        \partcaption{0.75}{-4.5}{Case 1}

        \STAB{3.5}{0}{T}{4}{4}{2}{$T_1$}
        \STAB{5}{0}{T}{4}{5}{2}{$T_2$}
        \partcaption{4.25}{-4.5}{Case 2}
        
        \STAB{0}{-8.5}{T}{5}{4}{2}{$T_1$}
        \STAB{1.5}{-8.5}{T}{5}{5}{2}{$T_2$}
        \partcaption{0.75}{-13}{Case 3}
        
    \end{tikzpicture}
}

\newcommand{\figoddb}[2]{
    \begin{tikzpicture}
        \STAB{0}{0}{B}{#1}{4}{2}{$B_1$}
        \STAB{2}{0}{B}{#1}{5}{2}{$B_2$}
        \STAB{4}{0}{B}{#1}{3}{2}{$B_1 \cup B_2$}
        \partcaption{2}{-4.5}{Case #2}
    \end{tikzpicture}
}

\newcommand{\argmin}{\mathop{\rm argmin}\limits}

\newcommand{\opt}{\mathop{\rm opt}}
\newcommand{\wor}{\mathop{\rm wor}}
\newcommand{\apx}{{\rm apx}}

\usepackage{here}

\usepackage{usebib}

\usepackage{setspace}

\usepackage{layout}
\begin{document}

\title{A 3/4 Differential Approximation Algorithm \\ for Traveling Salesman Problem}

\author[1]{Yuki Amano}
\author[1]{Kazuhisa Makino}
\affil[1]{Research Institute for Mathematical Sciences, Kyoto University, Kyoto, Japan. \texttt{\{ukiamano,makino\}@kurims.kyoto-u.ac.jp}}

\date{}
\maketitle
\begin{abstract}
In this paper, we consider differential approximability of the traveling salesman problem (TSP). 
We show that TSP is $3/4$-differential approximable, which improves the currently best known bound $3/4 -O(1/n)$ due to Escoffier and Monnot in 2008, 
where $n$ denotes the number of vertices in the given graph.

\end{abstract} 

\section{Introduction}
	The traveling salesman problem (TSP) finds a shortest \emph{Hamiltonian cycle} in a given complete graph with edge length, when a cycle is called \emph{Hamiltonian} (also called a \emph{tour}) if it visits every vertex exactly once. TSP is one of the most fundamental NP-hard optimization problems in operations research and computer science, and has been intensively studied from both practical and theoretical view points \cite{cook2011pursuit,monnot2014traveling,punnen2007traveling,shmoys1985traveling}. It has a number of applications such as planning, logistics, and the manufacture of microchips \cite{bland1989large,grotschel1991optimal}.
	Because of these importance, many heuristics and exact algorithms have been proposed \cite{bellman1961dynamic,held1962dynamic,lin1973effective,little1963algorithm}. From a view point of computational complexity, TSP is NP-hard, even in the Euclidean case, which includes the metric case. It is known that metric TSP is approximable with factor $1.5$ \cite{christofides1976worst}, and inapproximable with factor $117/116$ \cite{chlebik2019approximation}. Euclidean TSP admits a polynomial-time approximation scheme (PTAS), if the dimension of the Euclidean space is bounded by a constant \cite{arora1998polynomial}. We note that the approximation factors (i.e., ratios) above are widely used to analyze approximation algorithms.
	
	Let $\Pi$ be an optimization problem, and let $I$ be an instance of $\Pi$. We denote by $\opt(I)$ the value of an optimal solution to $I$. For an approximation algorithm $A$ for $\Pi$, we denote by ${\textstyle \apx_{A}}(I)$ the value of the approximate solution computed by $A$ for the instance $I$. Let \[r_A(I)=\apx_A(I)/\opt(I), \] and define the \emph{standard approximation ratio} of $A$ by $\sup_{I \in \Pi}r_A(I)$, where we assume that $\Pi$ is a minimization problem. Although the standard approximation ratio is well-studied and an important concept in algorithm theory, it is not invariant under affine transformation of the objective function. Namely, if the objective function $f(x)$ is replaced by $a+bf(x)$ for some constant $a$ and $b$, which might depend on the instance $I$, the standard ratio is not preserved. For example, the vertex cover problem and the independent set problem have affinely dependent objective functions. However they have different characteristics in the standard approximation ratio. The vertex cover problem is $2$-approximable \cite{papadimitriou1998combinatorial}, while the independent set problem is inapproximable within $O(n^{1-\epsilon})$ for any $\epsilon > 0$ \cite{engebretsen2000clique}, where $n$ denotes the number of vertices in a given graph. In order to remedy to this phenomenon, Demange and Paschos \cite{demange1996approximation} proposed the \emph{differential approximation ratio} defined by $\sup_{I \in \Pi} \rho_A(I)$, where 	
	\[\rho_A(I)=\frac{\wor(I)-\apx_A(I)}{\wor(I)-\opt(I)}\] and $\wor(I)$ denotes the value of a worst solution to $I$. Note that for any instance $I$ of $\Pi$ \[\apx_A(I) = \rho_A(I)\opt(I) + (1-\rho_A(I))\wor(I).\] Thus we have $0 \leq \rho_A(I) \leq 1$ and the larger $\rho_A(I)$ implies the better approximation for the instance $I$. Moreover, by definition, the differential approximation ratio remains invariant under affine transformation of the objective function. For this, it has been recently attracted much attention in approximation algorithm \cite{ausiello2005completeness}. It is known \cite{monnot2003approximation} that TSP, metric TSP, max TSP, and max metric TSP are affinely equivalent, i.e., their objective functions are transferred to each other by affine transformations, where max TSP is the problem to find a longest Hamiltonian cycle and max metric TSP is max TSP, in which the input weighted graph satisfies the metric condition. Therefore, these problems have the identical differential approximation ratio. 

	Hassin and Khuller \cite{hassin2001z} first studied differential approximability of TSP, and showed that it is $2/3$-differential approximable. Escoffier and Monnot \cite{escoffier2008better} improved it to $3/4-O(1/n)$, where $n$ denotes the number of vertices of a given graph. Monnot et al. \cite{monnot2002differential, monnot2003differential} showed that TSP is $3/4$-differential approximable if each edge length is restricted to one or two. 
	
	In this paper, we show that TSP is $3/4$-differential approximable, which inproves the currently best known results \cite{escoffier2008better,monnot2002differential, monnot2003differential}.
	Our algorithm is based on an idea in \cite{escoffier2008better} for the case in which a given graph $G$ has an even number of vertices and a triangle (i.e., cycle with $3$ edges\todo[fancyline]{of length three}) is contained in a minimum weighted $2$-factor of $G$. Their algorithm first computes minimum weighted $1$- and $2$-factors of a given graph, modify them to four path covers $P_i$ (for $i=1, \ldots , 4$), and then extend each path cover $P_i$ to a tour by adding edge set $F_i$ to it in such a way that at least one of the tours guarantees $3/4$-differential approximation ratio. Here the definitions of factor and path cover can be found in Section 2. We generalize their idea to the general even case. Note that $\bigcup_{i=1, \ldots ,4} F_i$ in their algorithm always forms a tour, where in general it does not. We show that there exists a way to construct path covers such that the length of $\bigcup_{i=1, \dots ,4} F_i$ is at most the worst tour length. Our algorithm for odd case is much more involved. For each path with three edges\todo[fancyline]{of length 3}, we first construct a $2$-factor and two path covers of a given graph which has minimum length among all these which completely and partially contains the path, modify them to eight path covers, and then extend each path cover to a tour, in such a way that at least one of the eight tours guarantees $3/4$-differential approximation ratio.
	
	The rest of the paper is organized as follows. In Section 2, we define basic concepts of graphs and discuss some properties on 2-matchings, which will be used in the subsequent sections. In Sections 3 and 4, we provide an approximation algorithms for TSP in which a given graph $G$ has even and odd numbers of vertices, respectively.

\section{Preliminary}
	Let $G=(V,E)$ be an undirected graph, where $n$ and $m$ denote the number of vertices and edges in $G$, respectively. 
	In this paper, we assume that a given graph $G$ of TSP is complete, i.e., $E={{V}\choose {2}}$, and it has an edge length function $\ell:E \to \mathbb{R}_+$, where $\mathbb{R}_+$ denotes the set of nonnegative reals.
	For a set $F \subseteq E$, let $V(F)$ denote the set of vertices with incident edges in $F$, i.e., $V(F) = \{v \in V \mid \exists (v,w) \in F\}$. 
	A set $F \subseteq E$ is called \emph{spanning} if $V(F)=V$, and \emph{acyclic} if $F$ contains no cycle. For a positive integer $k$, a set $F \subseteq E$ is called a \emph{$k$-matching} (resp., \emph{$k$-factor}) if each vertex has at most (resp., exactly) $k$ incident edges in $F$. Here $1$-matching is simply called a \emph{matching}. Note that an acyclic $2$-matching $F$ corresponds to a family of vertex-disjoint paths denoted by $\mathcal{P}(F) \subseteq 2^E$. A $2$-matching is called a \emph{path cover} if it is spanning and acyclic. For a set $F \subseteq E$, $V_1(F)$ and $V_2(F)$ respectively denote the sets of vertices with one and two incident edges in $F$. For a set $F \subseteq E$ and a vertex $v \in V$, let $\delta_F(v) = \{e \in F\mid e \text{ is incident to } v\}$.
	\begin{dfn}\label{dfn:pair-property}
		A pair of spanning 2-matchings $(S,T)$ is called \emph{valid} if it satisfies the following three conditions:
	\begin{align}
		(\mathrm{i})&\ T \text{ is acyclic}. \notag \\
		(\mathrm{ii})&\ \delta_S(v) = \delta_T(v) \text{ for any } v \in V_2(S) \cap V_2(T). \label{dfn:pair-property:edge}\\
		(\mathrm{iii})&\ V(C) \not = V(P) \text{ for any cycle } C \subseteq S \text{ and any path } P \subseteq \mathcal{P}(T). \label{dfn:pair-property:path}
	\end{align}
	\end{dfn}
	Figure 1 shows a valid pair of spanning $2$-matchings. 
	\begin{lem}
		\label{lem:movable-edges}
		Let $(S,T)$ be a valid pair of spanning $2$-matchings. If $S$ contains a cycle $C$, then $C$ contains two edges $e_i$ for $i=1,2$ such that $S_i=S\setminus \{e_i\}$ and $T_i=S\cup \{e_i\}$ satisfy the following three conditions:
		\begin{align}
			&(S_i, T_i) \text{ is valid for } i=1,2. \label{lem:movable-edges:valid}\\
			&V_1(S_i) \cup V_1(T_i) = V_1(S) \cup V_1(T) \text{ and } V_1(S_i) \cap V_1(T_i) = V_1(S) \cap V_1(T) \text{ for } i = 1,2. \label{lem:movable-edges:vertex}\\
			&\mathcal{P}(T) \text{ contains a path } P \text{ such that } P \cup \{e_1\} \text{ and } P \cup \{e_2\} \text{ are both paths. }\label{lem:movable-edges:path}
		\end{align}
	\end{lem}
	\begin{proof}
		Let $C=\{(v_0,v_1),(v_1,v_2), \ldots , (v_{k-1},v_k)\}$ for $k\geq 3$, where $v_k=v_0$. If $\mathcal{P}(T)$ contains a $(s,t)$-path $P$ with $s \in V(C)$ and $t \not\in V(C)$, then it follows from (\ref{dfn:pair-property:edge}) that $V(P) \cap V(C) = \{s\}$. We assume that $s = v_1$ without loss of generality. Let $e_1=(v_0,v_1)$ and $e_2=(v_1,v_2)$. It is not difficult to see that $(S_i, T_i)$ is valid and $P \cup \{e_i\}$ is a path for every $i=1,2$. On the other hand, if $\mathcal{P}(T)$ contains a $(s,t)$-path $P$ with $s,t \in V(C)$, we assume without loss of generality that $s=v_1$, $t=v_j$ for some $j$ with $2 \leq j \leq k-1$, and $P$ does not contain $(v_0,v_1)$. Note that such a path exists, since $T$ is spanning. Define $e_1=(v_0,v_1)$ and $e_2=(v_j,v_{j+1})$. By (\ref{dfn:pair-property:edge}) and (\ref{dfn:pair-property:path}), we can show that $(S_i, T_i)$ is valid and $P \cup \{e_i\}$ is a path for every $i=1,2$. This completes the proof. 
	\end{proof}
	
	Note that $(S_1,T_1)$ and $(S_2,T_2)$ in Lemma 2 satisfy
	\begin{align}
		S_i \cup T_i = S \cup T \text{ and } S_i \cap T_i = S \cap T \text{ for } i = 1,2,\label{lem:movable-edges:edge}
	\end{align}
	which immediately implies
	\begin{align}
		\ell(S_i) + \ell(T_i) = \ell(S) + \ell(T) \text{ for } i = 1,2, \label{lem:movable-edges:weight}
	\end{align}
	where $\ell(F)=\sum_{e \in F}\ell(e)$ for a set $F \subseteq E$.\\
	\textcolor{black}{Figure 2 shows two pairs $(S\setminus \{e_c\}, T \cup \{e_c\})$ and $(S\setminus \{e_d\}, T \cup \{e_d\})$ satisfying $(\ref{lem:movable-edges:valid})$, $(\ref{lem:movable-edges:vertex})$ and $(\ref{lem:movable-edges:path})$, which are obtained from $(S,T)$, $e_1=e_c$, and $e_2=e_d$ in Fig. \ref{fig:exeven}.}
	\begin{figure}
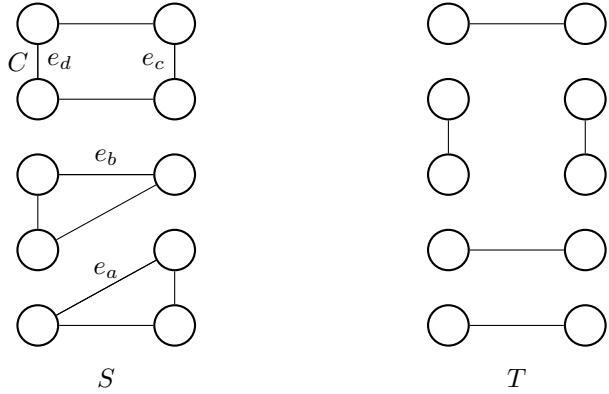

		\centering
		\figexeven
		\caption{A valid pair $(S,T)$ of spanning $2$-matchings.}\label{fig:exeven}
	\end{figure}
	\begin{figure}
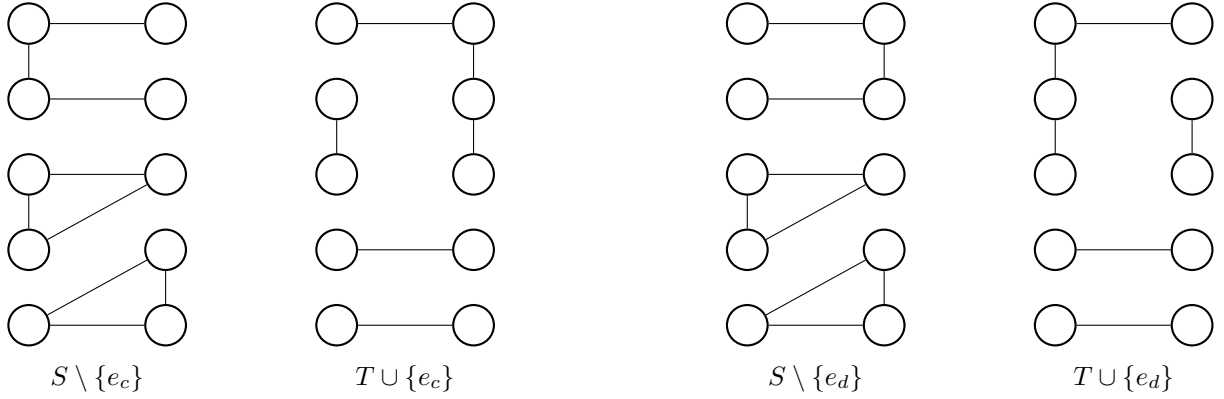

		\centering
		\figexevene
		\caption{Two pairs $(S\setminus \{e_c\}, T \cup \{e_c\})$ and $(S\setminus \{e_d\}, T \cup \{e_d\})$ satisfying $(\ref{lem:movable-edges:valid})$, $(\ref{lem:movable-edges:vertex})$ and $(\ref{lem:movable-edges:path})$, which are obtained from $(S,T)$, $e_1=e_c$, and $e_2=e_d$ in Fig. \ref{fig:exeven}.}\label{fig:exevene}
	\end{figure}
\section{Approximation for even instances}
	In this section, we construct an approximation algorithm for TSP in which a given graph has an even number of vertices. Our algorithm first construct four path covers from minimum weighted $1$- and $2$-factors of a given graph $G$, and then extend each path cover to a tour in such a way that at least one of the tours guarantees 3/4-differential approximation ratio.

Let us first describe the procedure \texttt{FourPathCovers}.
Let $(S,T)$ be a valid pair of spanning 2-matchings of $(G,\ell)$ such that $S$ is a 2-factor. 
The procedure computes from $(S,T)$ four path covers $S_1$, $S_2$, $T_1$, and $T_2$ that satisfies (\ref{lem:movable-edges:vertex}), (\ref{lem:movable-edges:edge}), $V_1(S_i)$ and $V_1(T_i)$ is a partition of $V_1(T)$ for $i=1,2$, i.e.,
	\begin{align}
		\begin{split}
		&V_1(S_i) \cup V_1(T_i) = V_1(T) \text{ and } V_1(S_i) \cap V_1(T_i) = \emptyset \text{ for } i=1,2,
		\end{split}\label{lem:4as2m:vertex}
	\end{align}
	and
	\begin{align}
		\begin{split}
		&\text{ there exist } e_1, e_2 \in E \text{ and } P \in \mathcal{P}(T_1 \cap T_2) \text{ such that }\\ &\quad T_1 \setminus T_2=\{e_1\},\ T_2 \setminus T_1=\{e_2\},\ P \cup \{e_1\} \in \mathcal{P}(T_1), \text{ and } P \cup \{e_2\} \in \mathcal{P}(T_2).
		\end{split}\label{lem:4as2m:diff}
	\end{align}
	\begin{figure}[h!t]
		\begin{procedure}[H]
			\caption{\texttt{FourPathCovers}$(S,T)$\\ /*$(S, T)$ is a valid pair of spanning $2$-matchings such that $S$ has a cycle. The procedure returns $4$ path covers $S_1$, $S_2$, $T_1$, and $T_2$ that sasifies (\ref{lem:movable-edges:vertex}), (\ref{lem:movable-edges:edge}), and (\ref{lem:4as2m:diff}).*/}
			\label{alg:FourPathCovers}
			\begin{algorithmic}
			\If{$S$ has exactly one cycle} 
				\State Take two edges $e_1$ and $e_2$ in Lemma \ref{lem:movable-edges}.
				\State \Return $S_1 = S \setminus \{e_1\}$, $T_1 = T \cup \{e_1\}$, $S_2 = S \setminus \{e_2\}$, and $T_2 = T \cup \{e_2\}$
			\Else \Comment{$S$ has at least two cycles.}
				\State Take an edge $e_1$ in Lemma \ref{lem:movable-edges}.
				\State \Return \texttt{FourPathCovers}$(S \setminus \{e_1\}, T \cup \{e_1\})$
			\EndIf
			\end{algorithmic}	
		\end{procedure}
	\end{figure}
	\textcolor{black}{In Fig. \ref{fig:exeven-st} we apply \textbf{Procedure} \texttt{FourPathCovers} to $(S,T)$ in Fig. \ref{fig:exeven}.}
	\begin{lem}\label{lem:4as2m}
	For a graph $G=(V,E)$, let $(S,T)$ be a valid pair of spanning $2$-matchings such that $S$ has a cycle. Then \textbf{Procedure} \texttt{FourPathCovers} returns four path covers $S_1$, $S_2$, $T_1$, and $T_2$ that satisfy $(\ref{lem:movable-edges:vertex})$, $(\ref{lem:movable-edges:edge})$, and $(\ref{lem:4as2m:diff})$. Furthermore, if $S$ is addition a $2$-factor of $G$, then the four path covers satisfy $(\ref{lem:4as2m:vertex})$.
	\end{lem}
	\begin{proof}
		By repeatedly applying Lemma \ref{lem:movable-edges} to $(S,T)$, we can see that four path covers $S_1$, $S_2$, $T_1$, and $T_2$ returned by \textbf{Procedure} \texttt{FourPathCovers} satisfy (\ref{lem:movable-edges:vertex}), (\ref{lem:movable-edges:edge}), and (\ref{lem:4as2m:diff}).
	Furthermore, if $S$ is a $2$-factor of $G$, we have (\ref{lem:4as2m:vertex}), since $V_1(S) = \emptyset$.
	\end{proof}
	Note that $(S,T)$ is a valid and $V_1(S) \cup V_1(T) = V$, if $S$ and $T$ are $2$- and $1$-factor of $G$, respectively.
	\begin{figure}
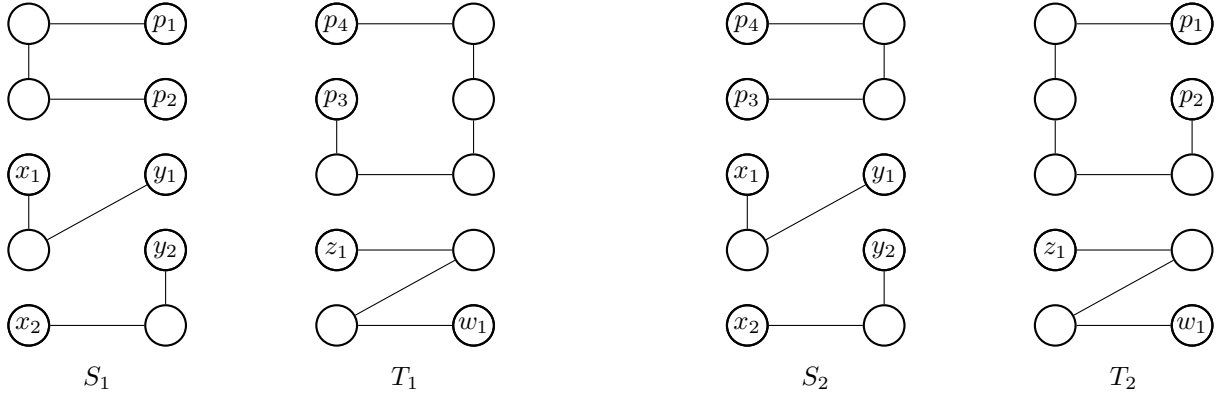

		\centering
		\figexevenst
		\caption{Two pairs $(S_1,T_1)$ and $(S_2,T_2)$ computed by \textbf{Procedure} \texttt{FourPathCovers} for a valid pair $(S,T)$, $e^{(1)}_1 = e_a$, $e^{(2)}_1 = e_b$, $e^{(3)}_1 = e_c$ and $e^{(3)}_2 = e_d$ in Fig. \ref{fig:exeven}, where $e^{(j)}_i$ denotes the edge chosen as $e_i$ in the $j$-th round of the procedure.}\label{fig:exeven-st}
	\end{figure}
	
	Let $S$ and $T$ be 2- and 1- factors of $G$, respectively. Note that our algorithm explain later makes use of minimum weighted 2-factor $S$ and 1-factor $T$ of $(G,\ell)$ which can be computed from $(G, \ell)$ in polynomial time. We assume that $S$ is not a tour of $G$, i.e., $S$ contains at least two cycles, since otherwise, $S$ itself is an optimal tour.
Let $S_1$, $S_2$, $T_1$, and $T_2$ be path covers returned by \textbf{Procedure} \texttt{FourPathCover}$(S,T)$.

	Let us then show how to construct edge sets $A_1$, $A_2$, $B_1$, and $B_2$, such that $S_i \cup A_i$ and $T_i \cup B_i$ (for $i = 1, 2$) are tours and $\ell(A_1)+\ell(A_2)+\ell(B_1)+\ell(B_2) \leq \wor(G,\ell)$, where $\wor(G,\ell)$ denotes the length of a longest tour of $(G,\ell)$. 

	Let $e_1$ = $(p_1, p_2)$ and $e_2=(p_3, p_4)$ be edges in Lemma \ref{lem:4as2m}. 
Since $e_1$ and $e_2$ are chosen from a cycle $C$, we can assume that 
$p_1 \not=p_3,p_4$ and $p_4 \not=p_1,p_2$, where $p_2 = p_3$ might hold.
We note that $\mathcal{P}(S_1) \setminus \mathcal{P}(S_2)$ consists of a $(p_1,p_2)$-path $P_1=C \setminus \{e_1\}$, and $\mathcal{P}(S_2) \setminus \mathcal{P}(S_1)$ consists of a $(p_3,p_4)$-path $P_2=C \setminus \{e_2\}$. 

Let $Q_i \, (i=1, \ldots, k)$ denote vertex-disjoint $(x_i,y_i)$-paths such that $\{Q_1, \dots , Q_k\} = \mathcal{P}(S_1) \cap \mathcal{P}(S_2)$ and $x_1$ and $y_1$ satisfy
\begin{align}
	\ell(p_2, x_1) + \ell(p_3, y_1) \leq \ell(p_2, y_1) + \ell(p_3, x_1). \label{xy-ineq}
\end{align}
\textcolor{black}{Figure \ref{fig:even-s} shows $S_1$ and $S_2$ computed by \textbf{Procedure} \texttt{FourPathCover}$(S,T)$, where two cases $p_2=p_3$ and $p_2\not=p_3$ are separately described.}
\begin{figure}
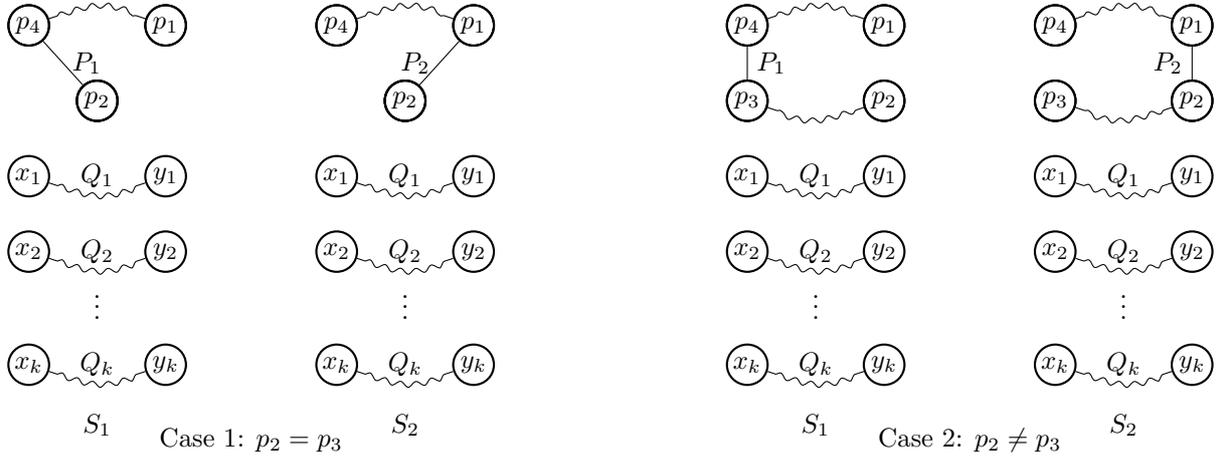

		\centering
		\figevens
		\caption{Two cases $p_2=p_3$ and $p_2\not=p_3$ for path covers $S_1$ and $S_2$ returned by \textbf{Procedure} \texttt{FourPathCovers}$(S,T)$.}\label{fig:even-s}
	\end{figure}
Define $A_1$ and $A_2$ by 
\begin{align}
	\begin{split}
	A_1 &= \{(p_2, x_1)\} \cup \{(y_i, x_{i+1}) \mid i = 1, \dots, k-1\} \cup \{(y_k, p_1)\}\\
	A_2 &= \{(p_3, y_1)\} \cup \{(x_i, y_{i+1}) \mid i = 1, \ldots, k-1\} \cup \{(x_k, p_4)\},
	\end{split}\label{addingA}
\end{align}
	\textcolor{black}{where the illustration can be found in Fig. \ref{fig:even-a}}. Then we have the following lemma.
	\begin{lem}\label{lem:addingA}
		Two sets $A_1$ and $A_2$ defined in $(\ref{addingA})$ satisfy the following three conditions.
		\begin{enumerate}
			\item[$(\mathrm{i})$] $S_i \cup A_i$ is a tour of $G$ for $i=1,2$.
			\item[$(\mathrm{ii})$] $V(A_i)=V_1(S_i)$ for $i=1,2$.
			\item[$(\mathrm{iii})$] $A_1 \cap A_2 = \emptyset$ and $A_1 \cup A_2$ consists of 
			\begin{enumerate}
				\item[$(\mathrm{iii}$-$1)$] a $(p_1, p_4)$-path if $p_2 = p_3$.
				\item[$(\mathrm{iii}$-$2)$] vertex-disjoint $(p_1, p_3)$- and $(p_2, p_4)$-paths if $p_2 \not= p_3$ and $k$ is odd. 
				\item[$(\mathrm{iii}$-$3)$] vertex-disjoint $(p_1, p_2)$- and $(p_3, p_4)$-paths if $p_2 \not= p_3$ and $k$ is even.
			\end{enumerate}
		\end{enumerate}
	\end{lem}
	\begin{proof}
		Note that $\mathcal{P}(S_1) = \{Q_1, \dots , Q_k\} \cup \{P_1\}$ and $\mathcal{P}(S_2) = \{Q_1, \dots , Q_k\} \cup \{P_2\}$. Thus it follows from the definition of $A_1$ and $A_2$.
	\end{proof}
	\begin{figure}
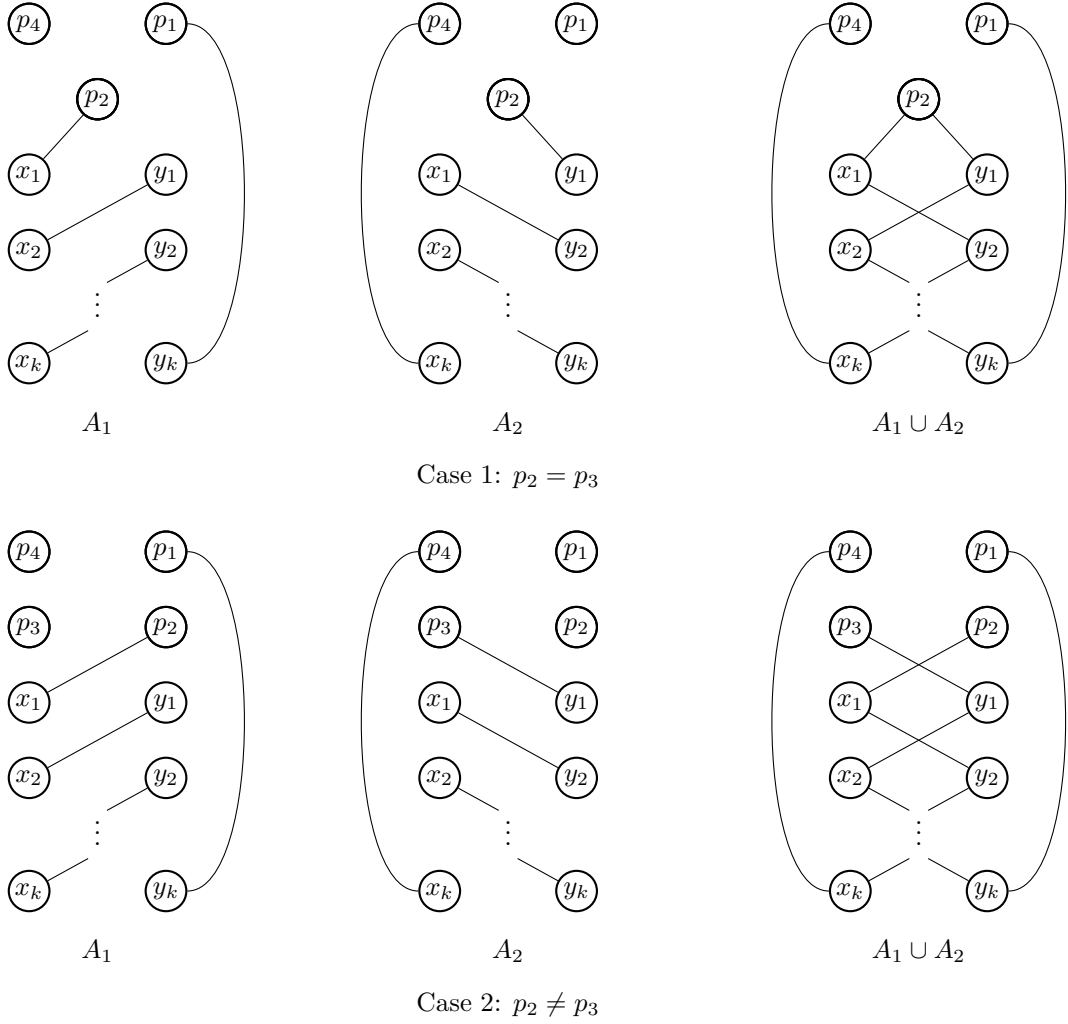

		\centering
		\figevena
		\caption{Two edge sets $A_1$ and $A_2$ for path covers $S_1$ and $S_2$ (as illustrated in Fig. \ref{fig:even-s}).}\label{fig:even-a}
	\end{figure}
	\begin{figure}
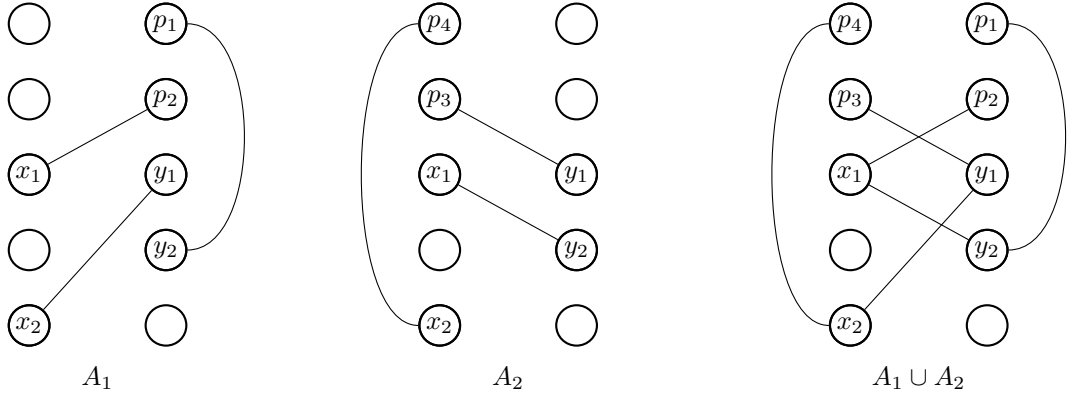

	\centering
		\figexevena
		\caption{Two edge sets $A_1$ and $A_2$ for path covers $S_1$ and $S_2$ in Fig. \ref{fig:exeven-st}.}\label{fig:exeven-a}
	\end{figure}
	\textcolor{black}{Figure \ref{fig:exeven-a} shows two edge sets $A_1$ and $A_2$ for $S_1$ and $S_2$ in Fig. \ref{fig:exeven-st}.}
	
	Let us next construct $B_1$ and $B_2$. Let $O_i \, (i=1,\ldots, d)$ denote\todo[fancyline]{odd に合わせてbe を denote に} vertex-disjoint $(z_i,w_i)$-paths such that $\{O_1, \ldots, O_d \} = \mathcal{P}(T_1) \cap \mathcal{P}(T_2)$. Note that $\mathcal{P}(T_1) \cap \mathcal{P}(T_2)=\emptyset$ (i.e., $d=0$) might hold. We separately consider the following four cases\textcolor{black}{, where the illustration can be found in Fig. \ref{fig:even-t}}. 
	\begin{enumerate}
		\item $p_2=p_3$ and $\mathcal{P}(T_1 \cap T_2)$ contains a $(p_1,p_4)$-path.
		\item $p_2=p_3$ and $\mathcal{P}(T_1 \cap T_2)$ contain no $(p_1,p_4)$-path.
		\item $p_2\not=p_3$ and $\mathcal{P}(T_1 \cap T_2)$ contains $(p_1,p_4)$- and $(p_2,p_3)$-paths.
		\item $p_2\not=p_3$ and $\mathcal{P}(T_1 \cap T_2)$ contains a $(p_2,p_3)$-path and no $(p_1,p_4)$-path.
	\end{enumerate}
	Here we recall that $e_1=(p_1,p_2)$ and $e_2=(p_3,p_4)$ satisfy Lemma 3.\medskip
	\begin{figure}
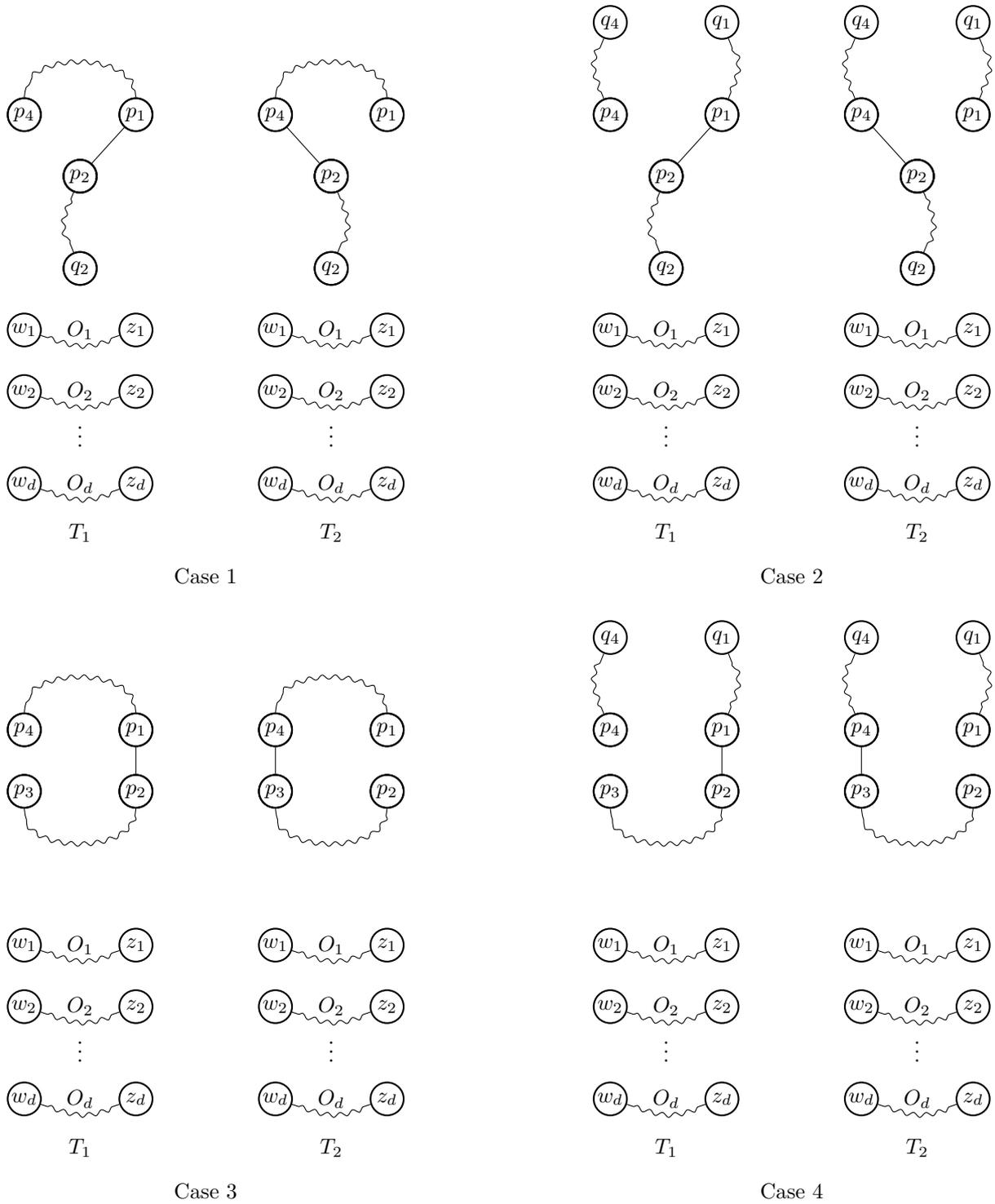

		\centering
		\figevent
		\caption{Four cases for path covers $T_1$ and $T_2$ returned by \textbf{Procedure} \texttt{FourPathCovers}$(S,T)$.}\label{fig:even-t}
	\end{figure}
	
	\textbf{Case 1}: Let $R_1$ denote a $(p_1,p_4)$-path in $\mathcal{P}(T_1 \cap T_2)$, and for some vertex $q_2$, let $R_2$ denote $(p_2, q_2)$-path in $\mathcal{P}(T_1 \cap T_2)$. 
Then\todo[fancyline]{By definition から変更(各Case)}, we have 
	\begin{align*}
		\mathcal{P}(T_1)&= \{O_1, \ldots, O_d\} \cup \{R_1 \cup \{ e_1\} \cup R_2\}\\
		\mathcal{P}(T_2)&= \{O_1, \ldots, O_d\} \cup \{R_1 \cup \{ e_2\} \cup R_2\},
	\end{align*}
	where $R_1\cup \{e_1\} \cup R_2$ and $R_1\cup \{e_2\} \cup R_2$ are $(p_4,q_2)$- and $(p_1,q_2)$-paths, respectively. Define $B_1$ and $B_2$ by 
	\begin{align}
		\begin{split}
		B_1 &= \begin{cases}
			\{(q_2,p_4)\} & \text{if } d=0\\
			\{(q_2,z_1)\} \cup \{(w_i,z_{i+1}) \mid i = 1, \ldots, d-1\} \cup \{(w_d, p_4)\} & \text{if } d \geq 1
				\end{cases}\\
		B_2 &= \begin{cases}
			\{(q_2,p_1)\} & \text{if } d=0\\
			\{(q_2,w_1)\} \cup \{(z_i,w_{i+1}) \mid i = 1, \ldots, d-1\} \cup \{(z_d, p_1)\} & \text{if } d \geq 1,
				\end{cases}
		\end{split}\label{addingB-1}
	\end{align}
	\textcolor{black}{as illustrated in Fig. \ref{fig:even-b-1}}.
	By definition, we have 
	\todo[fancyline]{Hamilton cycle を tourに}
	\begin{align}
		&T_i \cup B_i \text{ is a tour of } G \text{ for } i = 1,2, \label{addingB:Ham}\\
		&B_1 \cap B_2 = \emptyset \text{ and } V(B_i) = V_1(T_i) \text{ for } i = 1,2, \text{and} \label{addingB:vertex} \\
		&B_1 \cup B_2 \text{ is a $(p_1, p_4)$-path.} \label{addingB:path}
	\end{align} \medskip
	\begin{figure}
		\centering
		\figevenb[1]
		\caption{Two edge sets $B_1$ and $B_2$ for Case 1 (as illustrated in Fig. \ref{fig:even-t}).}\label{fig:even-b-1}
	\end{figure}
	
	\textbf{Case 2}: For some vertices $q_1,q_2$ and $q_4$, let $R_1$, $R_2$ and $R_4$ respectively denote $(p_1,q_1)$-, $(p_2, q_2)$-, and $(p_4,q_4)$-paths in $\mathcal{P}(T_1 \cap T_2)$. Then\todo[fancyline]{}, we have 
	\begin{align*}
		\mathcal{P}(T_1)&= \{O_1, \dots, O_d\} \cup \{R_4, R_1\cup \{e_1\} \cup R_2\}\\
		\mathcal{P}(T_2)&= \{O_1, \dots, O_d\} \cup \{R_1, R_4\cup \{e_2\} \cup R_2\},
	\end{align*}
	where $R_1\cup \{e_1\} \cup R_2$ and $R_4\cup \{e_2\} \cup R_2$ are $(q_1,q_2)$- and $(q_4,q_2)$-paths, respectively. Define $B_1$ and $B_2$ by 
	\begin{align}
		\begin{split}
		B_1 &= \begin{cases}
			\{(q_2,q_4),(p_4,q_1)\} & \text{if } d=0\\
			\{(q_2,z_1)\} \cup \{(w_i,z_{i+1}) \mid i = 1, \ldots, d-1\} \cup \{(w_d, q_4), (p_4,q_1)\} & \text{if } d \geq 1
				\end{cases}\\
		B_2 &= \begin{cases}
			\{(q_2,q_1),(p_1,q_4)\} & \text{if } d=0\\
			\{(q_2,w_1)\} \cup \{(z_i,w_{i+1}) \mid i = 1, \ldots, d-1\} \cup \{(z_d, q_1), (p_1,q_4)\} & \text{if } d \geq 1,
				\end{cases}
		\end{split}\label{addingB-2}
	\end{align}
	\textcolor{black}{as illustrated in Fig. \ref{fig:even-b-2}}.
	Similarly to Case 1, we have (\ref{addingB:Ham}), (\ref{addingB:vertex}) and (\ref{addingB:path}).\medskip
	
	\begin{figure}
		\centering
		\figevenb[2]
		\caption{Two edge sets $B_1$ and $B_2$ for Case 2 (as illustrated in Fig. \ref{fig:even-t}).}\label{fig:even-b-2}
	\end{figure}

	\textbf{Case 3}: Let $R_1$ and $R_2$ respectively denote $(p_1,p_4)$- and $(p_2,p_3)$-paths in $\mathcal{P}(T_1 \cap T_2)$. 
Then\todo[fancyline]{}, we have 
	\begin{align*}
		\mathcal{P}(T_1)&= \{O_1, \dots, O_d\} \cup \{R_1\cup \{e_1\} \cup R_2\}\\
		\mathcal{P}(T_2)&= \{O_1, \dots, O_d\} \cup \{R_1\cup \{e_2\} \cup R_2\},
	\end{align*}
	where $R_1\cup \{e_1\} \cup R_2$ and $R_1\cup \{e_2\} \cup R_2$ are
	$(p_3,p_4)$- and $(p_1,p_2)$-paths, respectively. Define $B_1$ and $B_2$ by 
	\begin{align}
		\begin{split}
		B_1 &= \begin{cases}
			\{(p_3,p_4)\} & \text{if } d=0\\
			\{(p_3,z_1)\} \cup \{(w_i,z_{i+1}) \mid i = 1, \ldots, d-1\} \cup \{(w_d, p_4)\} & \text{if } d \geq 1
				\end{cases}\\
		B_2 &= \begin{cases}
			\{(p_2,p_1)\} & \text{if } d=0\\
			\{(p_2,w_1)\} \cup \{(z_i,w_{i+1}) \mid i = 1, \ldots, d-1\} \cup \{(z_d, p_1)\} & \text{if } d \geq 1,
				\end{cases}
		\end{split}\label{addingB-3}
	\end{align}
	\textcolor{black}{as illustrated in Fig. \ref{fig:even-b-3}}.
	Similarly to the previous cases, we have (\ref{addingB:Ham}) and (\ref{addingB:vertex}). 
	Furthermore, $B_1 \cup B_2$ consist of\todo[fancyline]{is を変更} vertex-disjoint $(p_1, p_2)$- and
$(p_3, p_4)$-paths if $d$ is even, and vertex-disjoint $(p_1, p_3)$- and $(p_2, p_4)$-paths if $d$ is odd.\medskip

	\begin{figure}
		\centering
		\figevenb[3]
		\caption{Two edge sets $B_1$ and $B_2$ for Case 3 (as illustrated in Fig. \ref{fig:even-t}).}\label{fig:even-b-3}
	\end{figure}

	\textbf{Case 4}: Let $R_2$ denote $(p_2,p_3)$-path in $\mathcal{P}(T_1 \cap T_2)$, and for some vertices $q_1$ and $q_4$, let $R_1$ and $R_4$ respectively denote $(p_1, q_1)$- and $(p_4,q_4)$-paths in $\mathcal{P}(T_1 \cap T_2)$. 
Then\todo[fancyline]{}, we have 
	\begin{align*}
		\mathcal{P}(T_1)&= \{O_1, \dots, O_d\} \cup \{R_4, R_1\cup \{e_1\} \cup R_2\}\\
		\mathcal{P}(T_2)&= \{O_1, \dots, O_d\} \cup \{R_1, R_4\cup \{e_2\} \cup R_2\},
	\end{align*}
	where $R_1\cup \{e_1\} \cup R_2$ and $R_4\cup \{e_2\} \cup R_2$ are
	$(q_1,p_3)$- and $(q_4,p_2)$-paths, respectively. Define $B_1$ and $B_2$ by 
	\begin{align}
		\begin{split}
		B_1 &= \begin{cases}
			\{(p_3,q_4),(p_4,q_1)\} & \text{if } d=0\\
			\{(p_3,z_1)\} \cup \{(w_i,z_{i+1}) \mid i = 1, \ldots, d-1\} \cup \{(w_d, q_4), (p_4,q_1)\} & \text{if } d \geq 1
				\end{cases}\\
		B_2 &= \begin{cases}
			\{(p_2,q_1),(p_1,q_4)\} & \text{if } d=0\\
			\{(p_2,w_1)\} \cup \{(z_i,w_{i+1}) \mid i = 1, \ldots, d-1\} \cup \{(z_d, q_1), (p_1,q_4)\} & \text{if } d \geq 1,
				\end{cases}
		\end{split}\label{addingB-4}
	\end{align}
	\textcolor{black}{as illustrated in Fig. \ref{fig:even-b-4}}.
	Similarly to the previous cases, we have (\ref{addingB:Ham}) and (\ref{addingB:vertex}). 
	Furthermore, $B_1 \cup B_2$ consist of\todo[fancyline]{} vertex-disjoint $(p_1, p_3)$- and
$(p_2, p_4)$-paths if $d$ is even, and vertex-disjoint $(p_1, p_2)$- and $(p_3, p_4)$-paths if $d$ is odd. 

	\begin{figure}
		\centering
		\figevenb[4]
		\caption{Two edge sets $B_1$ and $B_2$ for Case 4 (as illustrated in Fig. \ref{fig:even-t}).}\label{fig:even-b-4}
	\end{figure}

	In summary, we have the following lemma. 
	\begin{lem}\label{lem:addingB}
		Let $B_1$ and $B_2$ be two edge sets defined as above. Then they satisfy $(\ref{addingB:Ham})$ and $(\ref{addingB:vertex})$, and $B_1 \cup B_2$ consists of $(\mathrm{i})$ a $(p_1,p_4)$-path if $q_2=q_3$, and either $(\mathrm{ii})$ vertex-disjoint $(p_1,p_2)$- and $(p_3,p_4)$-paths or $(\mathrm{iii})$ vertex-disjoint $(p_1,p_3)$- and $(p_2,p_4)$-paths if $q_2\not=q_3$.
	\end{lem}
	\todo[inline,fancyline, caption={Lemma5}]{修正前:Let $B_1$ and $B_2$ be two edge sets defined as above. Then they satisfy $(\ref{addingB:Ham})$ and $(\ref{addingB:vertex})$, and $B_1 \cup B_2$ consists of either $(\mathrm{i})$ a $(p_1,p_4)$-path, $(\mathrm{ii})$ vertex-disjoint $(p_1,p_2)$- and $(p_3,p_4)$-paths, or $(\mathrm{iii})$ vertex-disjoint $(p_1,p_3)$- and $(p_2,p_4)$-paths.}
	
	\begin{figure}
	\centering
		\figexevenb
		\caption{Two edge sets $B_1$ and $B_2$ for path covers $T_1$ and $T_2$ in Fig. \ref{fig:exeven-st}.}\label{fig:exeven-b}
	\end{figure}
	\textcolor{black}{Figure \ref{fig:exeven-b} shows two edge sets $B_1$ and $B_2$ for path covers $T_1$ and $T_2$ in Fig. \ref{fig:exeven-st}.}
	
	Furthermore, $A_i$ and $B_i$ ($i=1,2$) satisfy the following properties.
	\begin{lem}\label{lem:addingAB}
		Let $A_1$, $A_2$, $B_1$, and $B_2$ be defined as above. Then they are all pairwise disjoint, and $C = A_1 \cup A_2 \cup B_1 \cup B_2$ is a 2-factor, consisting of either one or two cycles. Furthermore, there exists a tour $H$ of $G$ such that $\ell(H) \geq \ell(C)$. 
	\end{lem}
	\todo[inline,fancyline, caption={Lemma6}]{修正前:Let $A_1$, $A_2$, $B_1$, and $B_2$ be defined as above. Then they are all pairwise disjoint, and $A_1 \cup A_2 \cup B_1 \cup B_2$ is a 2-factor, consisting of either one or two cycles. Furthermore, if it consists of two cycles, then there exists a tour $H$ in $G$ such that $\ell(H) \geq \ell(A_1) + \ell(A_2)+\ell(B_1)+\ell(B_2)$. }
	\begin{proof}
		It is not difficult to see that $A_1$, $A_2$, $B_1$, and $B_2$ are pairwise disjoint. Lemmas \ref{lem:4as2m}, \ref{lem:addingA}, and \ref{lem:addingB} imply that $C=A_1 \cup A_2 \cup B_1 \cup B_2$ is a 2-factor consisting of either one or two cycles. Thus if $C$ is a 2-factor, the latter statement of the lemma holds. Assume that $C$ consists of two cycles. In this case, we can see that two edges $(p_2,x_1)$ and $(p_3,y_1)$ belong to different cycles by (\ref{addingA}). Let $H=(C \setminus \{(p_2,x_1), (p_3,y_1)\}) \cup \{(p_2,y_1), (p_3,x_1)\}$ \textcolor{black}{(see in Fig. \ref{fig:cdh})}. Then $H$ is a tour of $G$. By assumption 
(\ref{xy-ineq}), we have $\ell(H)\geq\ell(C)$, which completes the proof. 
	\end{proof}
	\todo[inline,caption={Lemma6 proof}]{修正前:It is not difficult to see that $A_1$, $A_2$, $B_1$, and $B_2$ are pairwise disjoint. 
Lemmas \ref{lem:4as2m}, \ref{lem:addingA}, and \ref{lem:addingB} imply that $C=A_1 \cup A_2 \cup B_1 \cup B_2$ is a 2-factor consisting of either one or two cycles. Let us then consider the case in which $C$ consists of two cycles. 
In this case, we can see that two edges $(p_2,x_1)$ and $(p_3,y_1)$ belong to different cycles.
Let $H=(C \setminus \{(p_2,x_1), (p_3,y_1)\}) \cup \{(p_2,y_1), (p_3,x_1)\}$. Then $H$ is a cycle (i.e., a tour), and by assumption 
(\ref{xy-ineq}), we have $\ell(H)\geq\ell(C)$, which completes the proof. }

	\begin{figure}
		\centering
		\CDH
		\caption{A tour $H$ in the proof of Lemma \ref{lem:addingAB}, when $C$ consists of two cycles.}\label{fig:cdh}
	\end{figure}

We are now ready to describe our approximation algorithm. 
	\begin{figure}[h!t]
		\begin{algorithm}[H]
			\caption{\texttt{TourEven}}
			\label{alg:TourEven}
			\begin{algorithmic}
			\Require A complete graph $G=(V,E)$ with even $|V|$, and an edge length function $\ell :E \to \mathbb{R}_+$.
			\Ensure A tour $T_{\apx}$ in $G$.\vspace{3pt}
			\State Compute minimum weighted 2-factor $S$ and 1-factor $T$ of $(G,\ell)$.
			\If{$S$ is a tour}
				\State $T_{\apx} := S$.
			\Else
				\State $S_1, T_1, S_2, T_2 := \texttt{FourPathCovers}(S,T)$.\vspace{1pt}
				\State Compute edge sets $A_1$, $A_2$, $B_1$, $B_2$ defined in (\ref{addingA}), (\ref{addingB-1}), (\ref{addingB-2}), (\ref{addingB-3}) and (\ref{addingB-4}).\vspace{1pt}
				\State $\mathcal{T} := \{S_1 \cup A_1, S_2 \cup A_2, T_1 \cup B_1, T_2 \cup B_2 \}$.\vspace{1pt}
				\State $T_{\apx} := \argmin_{T \in \mathcal{T}} \ell(T)$.\vspace{1pt}
			\EndIf
			\State Outputs $T_{\apx}$ and halt.
			\end{algorithmic}
		\end{algorithm}
	\end{figure}
	
	\begin{thm}\label{thm:even-delta}
		For a complete graph $G=(V,E)$ with an even number of vertices and an edge length function $\ell: E \to R_+$, \textbf{Algorithm} \texttt{TourEven} computes a $3/4$-differential approximate tour of $(G,\ell)$ in polynomial time.
	\end{thm}
	\todo[inline, caption={Theorem7}]{修正前:Let $G = (V, E)$ be a complete graph with an even number of vertices, and $\ell: E \to R_+$ be an edge length function. Then \textbf{Algorithm} \texttt{TourEven} computes a $3/4$-differential approximate tour for TSP in polynomial time.}
	\begin{proof}
		We show that \textbf{Algorithm} \texttt{TourEven} outputs a 3/4-differential approximate tour $T_{\apx}$ in polynomial time. 
		If a minimum weighted 2-factor $S$ of $(G, \ell)$ computed in the algorithm is a tour, then clearly $T_{\apx}=S$ is an optimal tour.
On the other hand, if $S$ is not a tour, then we have 
		\begin{align*}
		4\ell(T_{\apx}) &\leq \ell(S_1 \cup A_1) + \ell(S_2 \cup A_2) + \ell(T_1 \cup B_1) + \ell(T_2 \cup B_2)\\
		&= 2(\ell(S)+\ell(T)) + \ell(A_1 \cup A_2 \cup B_1 \cup B_2)\\
		&\leq 3\opt(G, \ell) + \wor(G, \ell),
		\end{align*}
		where the first equality follows from Lemmas \ref{lem:addingA}, \ref{lem:addingB}, and \ref{lem:addingAB}, and the last inequality follows from Lemma \ref{lem:addingAB}, and $\ell(S) \leq \opt(G,\ell)$, and $2\ell(T) \leq \opt(G,\ell)$.
		Thus $T_{\apx}$ is a 3/4-differential approximate tour.
		Note that minimum weighted $1$- and $2$- factors can be computed in polynomial time, and 
$A_i$ and $B_i$ ($i=1,2$) can be computed in polynomial time.
Thus \textbf{Algorithm} \texttt{TourEven} is polynomial, which completes the proof.
	\end{proof}
	
Before concluding the section, let us remark that $3/4$-differential approximability is known for graph with an even number of vertices \cite{escoffier2008better}. Different from the algorithm in \cite{escoffier2008better}, ours is constructed in a uniform framework, which can further be extended to the odd case.

\section{Approximation for odd instances}
	In this section, we construct an approximation algorithm for TSP with an odd number of vertices. 
	Our algorithm is much more involved than the even case. It first guesses a path $P$ with three edges\todo[fancyline]{of length 3} in an optimal tour, constructs eight path covers based on $P$, and extend each path cover to a tour in such a way that at least one of the eight tours guarantees 3/4-differential approximation ratio. 

More precisely, for each path $P$ with three edges\todo[fancyline]{of length 3}, say, $P=\{(v_1,v_2),(v_2,v_3),(v_3,v_4)\}$ with all $v_i$'s distinct, let $S$ be a minimum weighted $2$-factor among those containing $P$, let $T$ be a minimum weighted path cover among those satisfying $(v_1, v_2), (v_2, v_3) \in T$ and $V_1(T)=V \setminus \{v_2\}$, and let $T^\prime$ be a minimum weighted path cover among those satisfying $(v_2, v_3), (v_3, v_4) \in T^\prime$ and $V_1(T^\prime)=V \setminus \{v_3\}$. Assume that $S$ is not a tour, i.e., it contains at least two cycles, since otherwise, is optimal, and hence ensures $3/4$-differential approximability if some optimal tour contains $P$. We note that $(S,T)$ and $(S,T^\prime)$ are both valid pairs of spanning 2-matchings. We apply \textbf{Procedure} \texttt{FourPathCovers} to them, but not arbitrarily. Let us specify two cycles $C^*$ and $C^{**}$ in $S$ such that $P \subseteq C^*$ and $P \cap C^{**} = \emptyset$. We define two vertices $v_0$ and $v_5$ in $V(C^*)$ such that $v_0\not=v_2$, $v_5\not=v_3$, and $(v_0, v_1) ,(v_4, v_5) \in C^*$. By definition $v_0=v_4$ and $v_5=v_1$ hold if $|C^*|=4$\todo[fancyline]{$C^*$ has length 4}. Furthermore, we define two edges $f$ and $f^\prime$ in $C^{**}$ that satisfy the properties in the next lemma.

\begin{lem}
	\label{lem:first-edge}
	Let $C^{**}$, $T$ and $T^\prime$ be defined as above. 
Then there exist two edges $f \in C^{**}\setminus T$ and $f^\prime \in C^{**}\setminus T^\prime$ such that
\begin{enumerate}
	\item[$(\mathrm{i})$] they have a common endpoint $q$, and
	\item[$(\mathrm{ii})$] $T \cup \{f\}$ and $T^\prime \cup \{f^\prime\}$ are path covers.
\end{enumerate}
\end{lem}

\begin{proof}
	If $C^{**} \setminus (T \cup T^\prime) \not= \emptyset$, then arbitrarily take an edge $f=f^\prime$ in $C^{**} \setminus (T \cup T^\prime)$. It is not difficult to see that (i) and (ii) in the lemma are satisfied. On the other hand, if $C^{**} \setminus (T \cup T^\prime) = \emptyset$. Then $C^{**}$ is even and it is covered with two matchings $C^{**} \cap T$ and $C^{**} \cap T^\prime$. This again implies the existence of two edges. 
\end{proof}
	\begin{figure}
		\centering
		\figexodd
		\caption{A 2-factor $S$ and two path covers $T$ and $T^\prime$ defined before Lemma \ref{lem:first-edge}, and an example of $f$ and $f^\prime$, which contain $q$.}\label{fig:exodd}
	\end{figure}

We note that $f$ and $f^\prime$ in Lemma \ref{lem:first-edge} might be identical, and (ii) in Lemma \ref{lem:first-edge} implies that two pairs $(S \setminus \{f\}, T \cup \{f\})$ and $(S \setminus \{f^\prime\}, T^\prime \cup \{f^\prime\})$ are valid. \textcolor{black}{Figure \ref{fig:exodd} shows an example of $S$, $T$, $T^\prime$, $f$ and $f^\prime$.}

Our algorithm uses \textbf{Procedure} \texttt{FourPathCovers} for $(S,T)$ defined as above 
in such a way that edge $e_1=f$ is chosen in the first round and two edges $e_1=(v_3,v_4)$ and $e_2=(v_0,v_1)$ are chosen in the last round. Similarly, our algorithm uses \textbf{Procedure} \texttt{FourPathCovers} for $(S,T^\prime)$ defined as above in such a way that edge $e_1=f^\prime$ is chosen in the first round and two edges $e_1=(v_1,v_2)$ and $e_2=(v_4,v_5)$ are chosen in the last round. Let $S_1$, $T_1$, $S_2$, and $T_2$ be four path covers obtained by \textbf{Procedure} \texttt{FourPathCover}$(S, T)$, and let $S^\prime_1$, $T^\prime_1$, $S^\prime_2$, and $T^\prime_2$ be four path covers returned by \textbf{Procedure} \texttt{FourPathCover}$(S, T^\prime)$.

\begin{lem}\label{lem:4as2m-odd}
	Let $S$, $T$, $S_i$, and $T_i\,(i=1,2)$ be defined as above. Then $S_1$, $S_2$, $T_1$, and $T_2$ are path covers such that 
	\begin{align}
		&S_i \cup T_i = S \cup T \text{ and } S_i \cap T_i = S \cap T \text{ for } i=1,2, \label{lem:4as2m-odd:edge}\\
		&V_1(S_i) \text{ and } V_1(T_i) \text{ is a partition of } V \setminus \{v_2\} \text{ for } i=1,2, \label{lem:4as2m-odd:vertex}\\
		&T_1 \setminus T_2 = \{(v_3,v_4)\}, T_2 \setminus T_1 = \{(v_0,v_1)\}, \text{and } \{(v_1,v_2), (v_2,v_3)\} \in \mathcal{P}(T_1 \cap T_2), \text{and}\label{lem:4as2-odd:path}\\
		&q \in V_1(S_1) \cap V_1(S_2),\label{lem:4as2m-odd:commonendpoint}
	\end{align}
	where $v_i \in V(C^*)\, (i=0,\ldots,4)$ are defined as above and $q$ is a common endpoint of $f$ and $f^\prime$ in \mbox{Lemma \ref{lem:first-edge}}.
\end{lem}
\begin{proof}
	By definition, we have $V_1(S) = \emptyset$ and $V_1(T) = V \setminus \{v_2\}$. Moreover, since an edge $f$ in Lemma \ref{lem:first-edge} is chosen in the first round of \textbf{Procedure} \texttt{FourPathCovers}$(S,T)$, and $(v_3,v_4)$ and $(v_0,v_1)$ are chosen in the last round of \textbf{Procedure} \texttt{FourPathCovers}$(S,T)$, Lemma \ref{lem:4as2m} implies the statement of lemma.
\end{proof}
	\textcolor{black}{Figure \ref{fig:exodd-st} shows $(S_1,T_1)$ and $(S_2, T_2)$ computed by \textbf{Procedure} \texttt{FourPathCovers} for $(S,T)$ in Fig. \ref{fig:exodd}.}
	\begin{figure}
		\centering
		\figexoddst
		\caption{Two pairs $(S_1,T_1)$ and $(S_2,T_2)$ computed by \textbf{Procedure} \texttt{FourPathCovers} for $(S,T)$, $e^{(1)}_1 = f$, $e^{(2)}_1 = e$, $e^{(3)}_1 = (v_3,v_4)$ and $e^{(3)}_2 = (v_0,v_1)$ in Fig. \ref{fig:exodd}, where $e^{(j)}_i$ denotes the edge chosen as $e_i$ in the $j$-th round of the procedure.}\label{fig:exodd-st}
	\end{figure}

Similarly, we have the following lemma.
\begin{lem}\label{lem:4as2m-odd'}
	Let $S$, $T^\prime$, $S^\prime_i$, and $T^\prime_i\,(i=1,2)$ be defined as above. Then $S^\prime_1$, $S^\prime_2$, $T^\prime_1$, and $T^\prime_2$ are path covers such that 
	\begin{align}
		&S^\prime_i \cup T^\prime_i = S \cup T^\prime \text{ and } S^\prime_i \cap T^\prime_i = S \cap T^\prime \text{ for } i=1,2, \label{lem:4as2m-odd':edge}\\
		&V_1(S^\prime_i) \text{ and } V_1(T^\prime_i) \text{ is a partition of } V \setminus \{v_3\} \text{ for } i=1,2, \label{lem:4as2m-odd':vertex}\\
		&T^\prime_1 \setminus T^\prime_2 = \{(v_1,v_2)\}, T^\prime_2 \setminus T^\prime_1 = \{(v_4,v_5)\}, \text{and } \{(v_2,v_3), (v_3,v_4)\} \in \mathcal{P}(T^\prime_1 \cap T^\prime_2), \text{and}\label{lem:4as2m-odd':path}\\
		&q \in V_1(S^\prime_1) \cap V_1(S^\prime_2),\label{lem:4as2m-odd':commonendpoint}
	\end{align}
	where $v_i \in V(C^*)\, (i=1,\ldots,5)$ are defined as above and $q$ is a common endpoint of $f$ and $f^\prime$ in \mbox{Lemma \ref{lem:first-edge}}.
\end{lem}
	\textcolor{black}{Figure \ref{fig:exodd-st'} shows $(S^\prime_1,T^\prime_1)$ and $(S^\prime_2, T^\prime_2)$ computed by \textbf{Procedure} \texttt{FourPathCovers} for $(S,T^\prime)$ in Fig. \ref{fig:exodd}.}
\begin{figure}
		\centering
		\figexoddstp
		\caption{Two pairs $(S^\prime_1,T^\prime_1)$ and $(S^\prime_2,T^\prime_2)$ computed by \textbf{Procedure} \texttt{FourPathCovers} for $(S,T^\prime)$, $e^{(1)}_1 = f^\prime$, $e^{(2)}_1 = e^\prime$, $e^{(3)}_1 = (v_1,v_2)$ and $e^{(3)}_2 = (v_4,v_5)$ in Fig. \ref{fig:exodd}, where $e^{(j)}_i$ denotes the edge chosen as $e_i$ in the $j$-th round of the procedure.}\label{fig:exodd-st'}
	\end{figure}

	Let us then show how to construct edge sets $A^{(\prime)}_i$ and $B^{(\prime)}_i$ (for $i=1,2$), such that $S^{(\prime)}_i \cup A^{(\prime)}_i$ and $T^{(\prime)}_i \cup B^{(\prime)}_i$ (for $i = 1,2$) are tours and \[\ell(A_1)+\ell(A_2)+\ell(B_1)+\ell(B_2)+\ell(A^\prime_1)+\ell(A^\prime_2)+\ell(B^\prime_1)+\ell(B^\prime_2) \leq 2\wor(G,\ell) - 2\ell(v_2,v_3),\]where $\wor(G,\ell)$ denotes the length of a longest tour of $(G,\ell)$. 
\begin{figure}
		\centering
		\figodds
		\caption{Two cases $|C^*| > 4$ and $|C^*| = 4$ for path covers $S_1$ and $S_2$ returned by \textbf{Procedure} \texttt{FourPathCovers}$(S,T)$.}\label{fig:odd-s}
	\end{figure}

Let us first show how to construct $A_1$ and $A_2$. By definition, $\mathcal{P}(S_1) \setminus \mathcal{P}(S_2)$ consists of a $(v_4,v_3)$-path $P_1 = C^* \setminus \{(v_4,v_3)\}$, and $\mathcal{P}(S_2) \setminus \mathcal{P}(S_1)$ consists of a $(v_1,v_0)$-path $P_2 = C^* \setminus \{(v_1,v_0)\}$. Let $Q_i \, (i=1,\ldots,k)$ denote $(x_i,y_i)$-paths such that $\{Q_1, \ldots, Q_k\} = \mathcal{P}(S_1) \cap \mathcal{P}(S_2)$, where $x_1=q$ in \mbox{Lemma \ref{lem:4as2m-odd}}. \textcolor{black}{Figure \ref{fig:odd-s} shows $S_1$ and $S_2$ computed by \textbf{Procedure} \texttt{FourPathCovers}$(S,T)$, where two cases $|C^*| > 4$ and $|C^*| = 4$ are separately described.} Define $A_1$ and $A_2$ by 
	\begin{align}
		\begin{split}
		A_1 &= \{(v_3, x_1)\} \cup \{(y_i, x_{i+1}) \mid i = 1, \dots, k-1\} \cup \{(y_k, v_4)\}\\
		A_2 &= \{(v_1, y_1)\} \cup \{(x_i, y_{i+1}) \mid i = 1, \ldots, k-1\} \cup \{(x_k, v_0)\},
		\end{split}\label{addingA-odd}
	\end{align}
	\textcolor{black}{as illustrated in Fig. \ref{fig:odd-a}}. Then we have the following lemma.
	\begin{figure}
		\centering
		\figodda
		\caption{Two edge sets $A_1$ and $A_2$ for path covers $S_1$ and $S_2$ (as illustrated in Fig. \ref{fig:odd-s}).}\label{fig:odd-a}
	\end{figure}
	
	\begin{lem}\label{lem:addingA-odd}
		Two sets $A_1$ and $A_2$ defined in $(\ref{addingA-odd})$ satisfy the following three conditions.
		\begin{enumerate}
			\item[$(\mathrm{i})$] $S_i \cup A_i$ is a tour of $G$ for $i=1,2$.
			\item[$(\mathrm{ii})$] $V(A_i)=V_1(S_i)$ for $i=1,2$.
			\item[$(\mathrm{iii})$] $A_1 \cap A_2 = \emptyset$ and $A_1 \cup A_2$ consists of 
			\begin{enumerate}
				\item[$(\mathrm{iii}$-$1)$] a $(v_1, v_3)$-path if $|C^*|=4$.
				\item[$(\mathrm{iii}$-$2)$] vertex-disjoint $(v_0, v_3)$- and $(v_1, v_4)$-paths if $|C^*|>4$ and $k$ is odd.
				\item[$(\mathrm{iii}$-$3)$] vertex-disjoint $(v_0, v_1)$- and $(v_3, v_4)$-paths if $|C^*|>4$ and $k$ is even.
			\end{enumerate}
		\end{enumerate}
	\end{lem}
	\begin{proof}
		Note that $\mathcal{P}(S_1) = \{Q_1, \dots , Q_k\} \cup \{P_1\}$ and $\mathcal{P}(S_2) = \{Q_1, \dots , Q_k\} \cup \{P_2\}$. Thus it follows from the definition of $A_1$ and $A_2$.
	\end{proof}
	\begin{figure}
	\centering
		\figexodda
		\caption{Two edge sets $A_1$ and $A_2$ for path covers $S_1$ and $S_2$ in Fig. \ref{fig:exodd-st}.}\label{fig:exodd-a}
	\end{figure}
	
Similarly, let us define $A^\prime_1$ and $A^\prime_2$. Recall that $\mathcal{P}(S^\prime_1) \setminus \mathcal{P}(S^\prime_2)$ consists of a $(v_1,v_2)$-path $P^\prime_1 = C^* \setminus \{(v_1,v_2)\}$, and $\mathcal{P}(S^\prime_2) \setminus \mathcal{P}(S^\prime_1)$ consists of a $(v_4,v_5)$-path $P^\prime_2 = C^* \setminus \{(v_4,v_5)\}$. Let $Q^\prime_i \, (i=1,\ldots,k)$ denote $(x^\prime_i,y^\prime_i)$-paths such that $\{Q^\prime_1, \ldots, Q^\prime_k\} = \mathcal{P}(S^\prime_1) \cap \mathcal{P}(S^\prime_2)$, where $x^\prime_1=q$ in \mbox{Lemma \ref{lem:4as2m-odd}}. Define $A^\prime_1$ and $A^\prime_2$ by 
	\begin{align}
		\begin{split}
		A^\prime_1 &= \{(v_2, x^\prime_1)\} \cup \{(y^\prime_i, x^\prime_{i+1}) \mid i = 1, \dots, k-1\} \cup \{(y^\prime_k, v_1)\}\\
		A^\prime_2 &= \{(v_4, y^\prime_1)\} \cup \{(x^\prime_i, y^\prime_{i+1}) \mid i = 1, \ldots, k-1\} \cup \{(x^\prime_k, v_5)\}.
		\end{split}\label{addingA-odd'}
	\end{align}
	Then we have the following lemma.
	\begin{lem}\label{lem:addingA-odd'}
		Two sets $A^\prime_1$ and $A^\prime_2$ defined in $(\ref{addingA-odd'})$ satisfy the following three conditions.
		\begin{enumerate}
			\item[$(\mathrm{i})$] $S^\prime_i \cup A^\prime_i$ is a tour of $G$ for $i=1,2$.
			\item[$(\mathrm{ii})$] $V(A^\prime_i)=V_1(S^\prime_i)$ for $i=1,2$.
			\item[$(\mathrm{iii})$] $A^\prime_1 \cap A^\prime_2 = \emptyset$ and $A^\prime_1 \cup A^\prime_2$ consists of 
			\begin{enumerate}
				\item[$(\mathrm{iii}$-$1)$] a $(v_2, v_4)$-path if $|C^*|=4$.
				\item[$(\mathrm{iii}$-$2)$] vertex-disjoint $(v_1, v_4)$- and $(v_2, v_5)$-paths if $|C^*|>4$ and $k$ is odd.
				\item[$(\mathrm{iii}$-$3)$] vertex-disjoint $(v_1, v_2)$- and $(v_4, v_5)$-paths if $|C^*|>4$ and $k$ is even.
			\end{enumerate}
		\end{enumerate}
	\end{lem}
	\begin{proof}
		Note that $\mathcal{P}(S^\prime_1) = \{Q^\prime_1, \dots , Q^\prime_k\} \cup \{P^\prime_1\}$ and $\mathcal{P}(S^\prime_2) = \{Q^\prime_1, \dots , Q^\prime_k\} \cup \{P_2^\prime\}$. Thus it follows from the definition of $A^\prime_1$ and $A^\prime_2$.
	\end{proof}
	\textcolor{black}{Figures \ref{fig:exodd-a} and \ref{fig:exodd-a'} show an example of edge sets $A_1$, $A_2$, $A^\prime_1$, and $A^\prime_2$ for path covers $S_1$, $S_2$, $S^\prime_1$, and $S^\prime_2$ in Figs. \ref{fig:exodd-st} and \ref{fig:exodd-st'}}
	\begin{figure}
	\centering
		\figexoddap
		\caption{Two edge sets $A^\prime_1$ and $A^\prime_2$ for path covers $S^\prime_1$ and $S^\prime_2$ in Fig. \ref{fig:exodd-st'}.}\label{fig:exodd-a'}
	\end{figure}
	\begin{figure}
		\centering
		\figoddt
		\caption{Three cases for path covers $T_1$ and $T_2$ returned by \textbf{Procedure} \texttt{FourPathCovers}$(S,T)$.}\label{fig:odd-t}
	\end{figure}
	Let us next construct $B_1$, $B_2$, $B^\prime_1$, and $B^\prime_2$. Let $O_i \, (i=1,\ldots,d)$ denote vertex-disjoint\todo[fancyline]{even に合わせてvertex-disjointを追加} $(z_i,w_i)$-paths such that $\{O_1, \ldots, O_d\} = \mathcal{P}(T_1) \cap \mathcal{P}(T_2)$, where $z_1$ and $w_1$ satisfy 
	\begin{align}
	\ell(v_1, z_1) + \ell(v_3, w_1) \leq \ell(v_1, w_1) + \ell(v_3, z_1). \label{zw-ineq}
	\end{align}
	We remark that $d \geq 1$ (i.e., $\mathcal{P}(T_1) \cap \mathcal{P}(T_2) \not = \emptyset$) holds if $n \geq 16$. To see this, we have $|\mathcal{P}(T_1)| = \lfloor n/2 \rfloor - (k+1)$, where $k+1$ is equal to the number of cycles in $S$. Since each cycle in $S$ has size at least $3$, $k+1 \geq \lfloor n/3 \rfloor$ holds, which implies that $|\mathcal{P}(T_1)| \geq 3$ if $n \geq 16$. Since $|\mathcal{P}(T_1) \setminus \mathcal{P}(T_2)| \leq 2$, we have $d=|\mathcal{P}(T_1) \cap \mathcal{P}(T_2)| \geq 1$ if $n\geq16$. 
	In the subsequent discussion, we assume that $n \geq 16$, and construct $B_1$ and $B_2$ by considering the following three cases \textcolor{black}{(see in Fig. \ref{fig:odd-t})}.
	\begin{enumerate}
		\item $|C^*|>4$ and $\mathcal{P}(T_1 \cap T_2)$ contains a $(v_0,v_4)$-path.
		\item $|C^*|>4$ and $\mathcal{P}(T_1 \cap T_2)$ contains no $(v_0,v_4)$-path.
		\item $|C^*|=4$.
	\end{enumerate}
	
	\textbf{Case 1}: Let $R_0$ denote a $(v_0,v_4)$-path in $\mathcal{P}(T_1 \cap T_2)$, and let $R_1 = \{(v_1,v_2),(v_2,v_3)\}$. By definition $R_1$ is a $(v_1, v_3)$-path in $\mathcal{P}(T_1 \cap T_2)$. Then we note that 
	\begin{align*}
		\mathcal{P}(T_1)&= \{O_1, \dots, O_d\} \cup \{R_0\cup \{(v_3,v_4)\} \cup R_1\}\\
		\mathcal{P}(T_2)&= \{O_1, \dots, O_d\} \cup \{R_0\cup \{(v_0,v_1)\} \cup R_1\},
	\end{align*}
	where $R_0\cup \{(v_3,v_4)\} \cup R_1$ and $R_0\cup \{(v_0,v_1)\} \cup R_1$ are $(v_0,v_1)$- and $(v_3,v_4)$-paths, respectively. Define $B_1$ and $B_2$ by 
	\begin{align}
		\begin{split}
		B_1 &= \{(v_1,z_1)\} \cup \{(w_i,z_{i+1}) \mid i = 1, \ldots, d-1\} \cup \{(w_d, v_0)\}\\
		B_2 &= \{(v_3,w_1)\} \cup \{(z_i,w_{i+1}) \mid i = 1, \ldots, d-1\} \cup \{(z_d, v_4)\},
		\end{split}\label{addingB-odd-1}
	\end{align}
	\textcolor{black}{as illustrated in Fig. \ref{fig:odd-b-1}}.
	Then we have
	\begin{align}
		&T_i \cup B_i \text{ is a tour of } G \text{ for } i = 1,2, \label{addingB-odd:Ham}\\
		&B_1 \cap B_2 = \emptyset \text{ and } V(B_i) = V_1(T_i) \text{ for } i = 1,2, \text{and} \label{addingB-odd:vertex} \\
		\begin{split}
		&B_1 \cup B_2 \text{ consist of vertex-disjoint } (v_0,v_1) \text{- and } (v_3,v_4)\text{-paths if } d \text{ is even,} \\ &\quad\text{and vertex-disjoint } (v_0,v_3) \text{- and } (v_1,v_4)\text{-paths if } d \text{ is odd.}
		\end{split}\label{addingB-odd:path}
	\end{align}\medskip
	\begin{figure}
		\centering
		\figoddb{3}{1}
		\caption{Two edge sets $B_1$ and $B_2$ for Case 1 (as illustrated in Fig. \ref{fig:odd-t}).}\label{fig:odd-b-1}
	\end{figure}
	
	\textbf{Case 2}: Let $R_1 = \{(v_1,v_2), (v_2,v_3)\}$ (i.e., let $R_1$ be a $(v_1,v_3)$-path in $\mathcal{P}(T_1 \cap T_2)$). Let $R_0$ and $R_4$ respectively denote $(v_0, r_0)$- and $(v_4,r_4)$-paths in $\mathcal{P}(T_1 \cap T_2)$. Then, we have 
	\begin{align*}
		\mathcal{P}(T_1)&= \{O_1, \dots, O_d\} \cup \{R_0, R_1\cup \{(v_3,v_4)\} \cup R_4\}\\
		\mathcal{P}(T_2)&= \{O_1, \dots, O_d\} \cup \{R_4, R_0\cup \{(v_0,v_1)\} \cup R_1\},
	\end{align*}
	where $R_1\cup \{(v_3,v_4)\} \cup R_4$ and $R_0\cup \{(v_0,v_1)\} \cup R_1$ are
	$(v_1,r_4)$- and $(v_3,r_0)$-paths, respectively. Define $B_1$ and $B_2$ by 
	\begin{align}
		\begin{split}
		B_1 &= \{(v_1,z_1)\} \cup \{(w_i,z_{i+1}) \mid i = 1, \ldots, d-1\} \cup \{(w_d, r_0), (v_0,r_4)\} \\
		B_2 &= \{(v_3,w_1)\} \cup \{(z_i,w_{i+1}) \mid i = 1, \ldots, d-1\} \cup \{(z_d, r_4), (v_4,r_0)\},
		\end{split}\label{addingB-odd-2}
	\end{align}
	\textcolor{black}{as illustrated in Fig. \ref{fig:odd-b-2}}.
	Similarly to Case 1, we have (\ref{addingB-odd:Ham}) and (\ref{addingB-odd:vertex}). 
	Furthermore, $B_1 \cup B_2$ consists of $(v_0, v_3)$- and $(v_1, v_4)$-paths if $d$ is even, and vertex-disjoint $(v_0, v_1)$- and $(v_3, v_4)$-paths if $d$ is odd. \medskip
	\begin{figure}
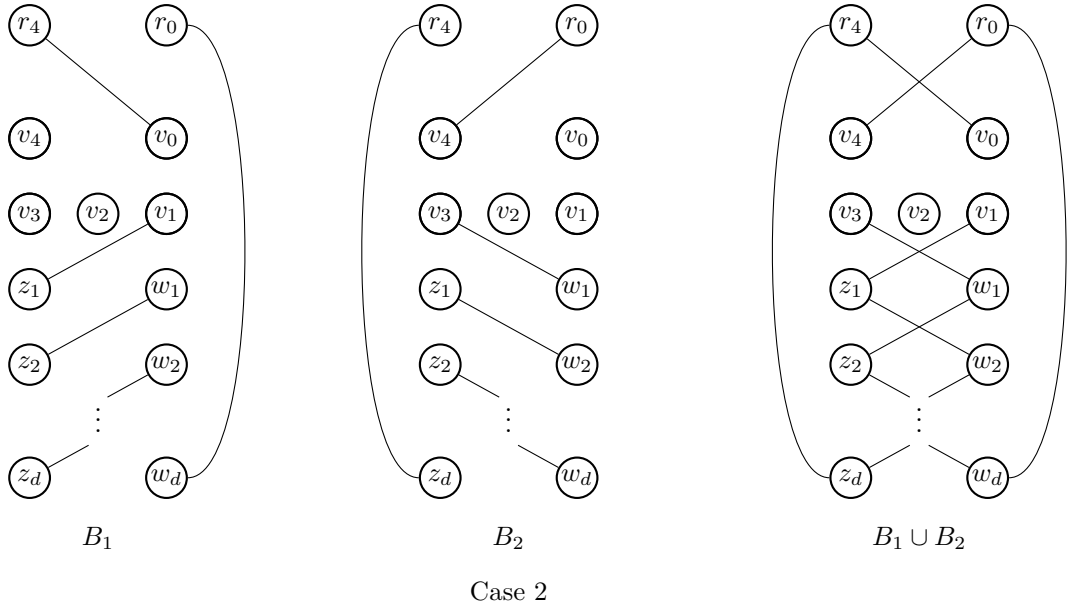

		\centering
		\figoddb{4}{2}
		\caption{Two edge sets $B_1$ and $B_2$ for Case 2 (as illustrated in Fig. \ref{fig:odd-t}).}\label{fig:odd-b-2}
	\end{figure}
	
	\textbf{Case 3}: In this case, we have $v_0=v_4$. Let $R_1 = \{(v_1,v_2), (v_2,v_3)\}$ (i.e., let be a $(v_1,v_3)$-path in $\mathcal{P}(T_1 \cap T_2)$), let $R_4$ denote $(v_4, r_4)$-path in $\mathcal{P}(T_1 \cap T_2)$. Then we have
	\begin{align*}
		\mathcal{P}(T_1)&= \{O_1, \dots, O_d\} \cup \{R_1\cup \{(v_3,v_4)\} \cup R_4\}\\
		\mathcal{P}(T_2)&= \{O_1, \dots, O_d\} \cup \{R_1\cup \{(v_1,v_4)\} \cup R_4\},
	\end{align*}
	where $R_1\cup \{(v_3,v_4)\} \cup R_4$ and $R_1\cup \{(v_1,v_4)\} \cup R_4$ are $(v_1,r_4)$- and $(v_3,r_4)$-paths, respectively. Define $B_1$ and $B_2$ by 
	\begin{align}
		\begin{split}
		B_1 &= \{(v_1,z_1)\} \cup \{(w_i,z_{i+1}) \mid i = 1, \ldots, d-1\} \cup \{(w_d, r_4)\} \\
		B_2 &= \{(v_3,w_1)\} \cup \{(z_i,w_{i+1}) \mid i = 1, \ldots, d-1\} \cup \{(z_d, r_4)\},
		\end{split}\label{addingB-odd-3}
	\end{align}
	\textcolor{black}{as illustrated in Fig. \ref{fig:odd-b-3}}.
	Similarly to the previous cases, we have (\ref{addingB-odd:Ham}) and (\ref{addingB-odd:vertex}). 
	Furthermore, we have $B_1 \cup B_2$ is a $(v_1,v_3)$-path.
	\begin{figure}
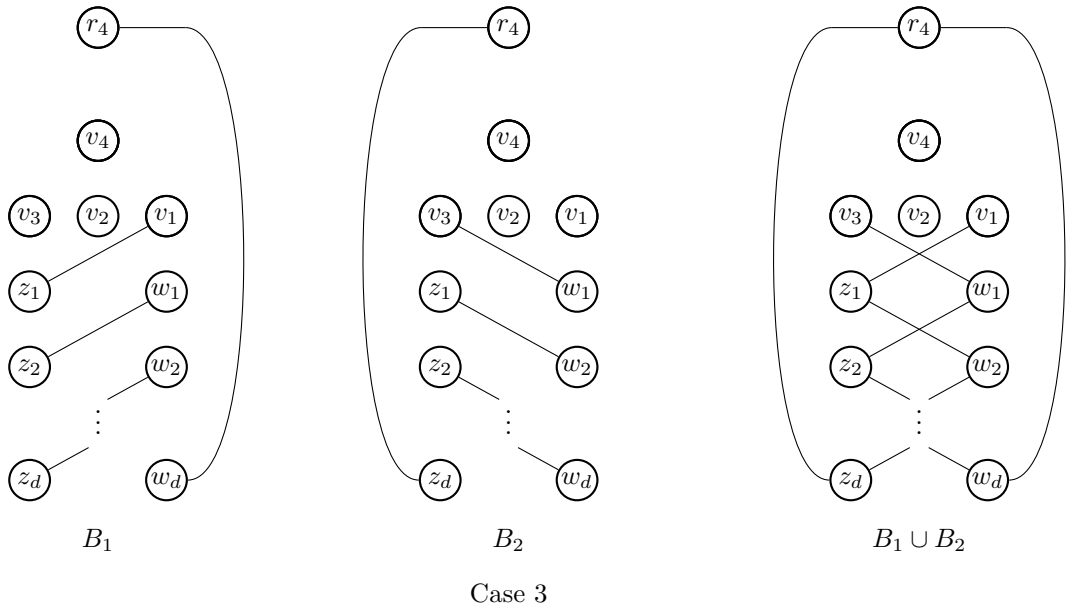

		\centering
		\figoddb{5}{3}
		\caption{Two edge sets $B_1$ and $B_2$ for Case 3 (as illustrated in Fig. \ref{fig:odd-t}).}\label{fig:odd-b-3}
	\end{figure}
	
	In summary, we have the following lemma.
	\begin{lem}\label{lem:addingB-odd}
		Let $B_1$ and $B_2$ be edge sets defined as above. Then they satisfy $(\ref{addingB-odd:Ham})$ and $(\ref{addingB-odd:vertex})$, and $B_1 \cup B_2$ consists of either $(\mathrm{i})$ vertex-disjoint $(v_0,v_3)$- and $(v_1,v_4)$-paths or $(\mathrm{ii})$ vertex-disjoint $(v_0,v_1)$- and $(v_3,v_4)$-paths if $|C^*|>4$, and $(\mathrm{iii})$ a $(v_1,v_3)$-path if $|C^*|=4$.
	\end{lem}
	
	Similarly, $B^\prime_1$ and $B^\prime_2$ can be obtained from $T^\prime_1$ and $T^\prime_2$ as follows. Let $O^\prime_i \, (i=1,\ldots,d)$ denote vertex-disjoint\todo[fancyline]{even に合わせてvertex-disjointを追加} $(z^\prime_i,w^\prime_i)$-paths such that $\{O^\prime_1, \ldots, O^\prime_d\} = \mathcal{P}(T^\prime_1) \cap \mathcal{P}(T^\prime_2)$, where $z^\prime_1$ and $w^\prime_1$ satisfy 
	\begin{align}
	\ell(v_4, z^\prime_1) + \ell(v_2, w^\prime_1) \leq \ell(v_4, w^\prime_1) + \ell(v_2, z^\prime_1). \label{zw-ineq'}
	\end{align}
	Recall that $d \geq 1$ (i.e., $\mathcal{P}(T^\prime_1) \cap \mathcal{P}(T^\prime_2) \not = \emptyset$) holds if $n \geq 17$. We construct $B^\prime_1$ and $B^\prime_2$ by considering the following three cases.
	\begin{enumerate}
		\item $|C^*|>4$ and $\mathcal{P}(T^\prime_1 \cap T^\prime_2)$ contains a $(v_1,v_5)$-path.
		\item $|C^*|>4$ and $\mathcal{P}(T^\prime_1 \cap T^\prime_2)$ contains no $(v_1,v_5)$-path.
		\item $|C^*|=4$.
	\end{enumerate}
	
	\textbf{Case 1}: Let $R^\prime_1$ denote a $(v_1,v_5)$-path in $\mathcal{P}(T^\prime_1 \cap T^\prime_2)$, and let $R^\prime_2 = \{(v_2,v_3),(v_3,v_4)\}$. By definition $R^\prime_2$ is a $(v_2, v_4)$-path in $\mathcal{P}(T^\prime_1 \cap T^\prime_2)$. Then we note that 
	\begin{align*}
		\mathcal{P}(T^\prime_1)&= \{O^\prime_1, \dots, O^\prime_d\} \cup \{R^\prime_1\cup \{(v_1,v_2)\} \cup R^\prime_2\}\\
		\mathcal{P}(T^\prime_2)&= \{O^\prime_1, \dots, O^\prime_d\} \cup \{R^\prime_1\cup \{(v_4,v_5)\} \cup R^\prime_2\},
	\end{align*}
	where $R^\prime_1\cup \{(v_1,v_2)\} \cup R^\prime_2$ and $R^\prime_1\cup \{(v_4,v_5)\} \cup R^\prime_2$ are $(v_4,v_5)$- and $(v_1,v_2)$-paths, respectively. Define $B^\prime_1$ and $B^\prime_2$ by 
	\begin{align}
		\begin{split}
		B^\prime_1 &= \{(v_4,z^\prime_1)\} \cup \{(w^\prime_i,z^\prime_{i+1}) \mid i = 1, \ldots, d-1\} \cup \{(w^\prime_d, v_5)\}\\
		B^\prime_2 &= \{(v_2,w^\prime_1)\} \cup \{(z^\prime_i,w^\prime_{i+1}) \mid i = 1, \ldots, d-1\} \cup \{(z^\prime_d, v_1)\}.
		\end{split}\label{addingB-odd'-1}
	\end{align}
	Then we have
	\begin{align}
		&T^\prime_i \cup B^\prime_i \text{ is a tour of } G \text{ for } i = 1,2, \label{addingB-odd':Ham}\\
		&B^\prime_1 \cap B^\prime_2 = \emptyset \text{ and } V(B^\prime_i) = V_1(T^\prime_i) \text{ for } i = 1,2, \text{and} \label{addingB-odd':vertex} \\
		\begin{split}
		&B^\prime_1 \cup B^\prime_2 \text{ consist of vertex-disjoint } (v_1,v_2) \text{- and } (v_4,v_5)\text{-paths if } d \text{ is even}, \\ &\quad\text{and vertex-disjoint } (v_1,v_4) \text{- and } (v_2,v_5)\text{-paths if } d \text{ is odd.}
		\end{split}\label{addingB-odd':path}
	\end{align}
	
	\textbf{Case 2}: Let $R^\prime_2 = \{(v_2,v_3), (v_3,v_4)\}$ (i.e., let be a $(v_2,v_4)$-path in $\mathcal{P}(T^\prime_1 \cap T^\prime_2)$). Let $R^\prime_1$ and $R^\prime_5$ respectively denote $(v_1, r^\prime_1)$- and $(v_5,r^\prime_5)$-paths in $\mathcal{P}(T^\prime_1 \cap T^\prime_2)$. Then, we have 
	\begin{align*}
		\mathcal{P}(T^\prime_1)&= \{O^\prime_1, \dots, O^\prime_d\} \cup \{R^\prime_5, R^\prime_1\cup \{(v_1,v_2)\} \cup R^\prime_2\}\\
		\mathcal{P}(T^\prime_2)&= \{O^\prime_1, \dots, O^\prime_d\} \cup \{R^\prime_1, R^\prime_2\cup \{(v_4,v_5)\} \cup R^\prime_5\},
	\end{align*}
	where $R^\prime_1\cup \{(v_1,v_2)\} \cup R^\prime_2$ and $R^\prime_2\cup \{(v_4,v_5)\} \cup R^\prime_5$ are
	$(v_4,r^\prime_1)$- and $(v_2,r^\prime_5)$-paths, respectively. Define $B^\prime_1$ and $B^\prime_2$ by 
	\begin{align}
		\begin{split}
		B^\prime_1 &= \{(v_4,z^\prime_1)\} \cup \{(w^\prime_i,z^\prime_{i+1}) \mid i = 1, \ldots, d-1\} \cup \{(w^\prime_d, r^\prime_5), (v_5,r^\prime_1)\} \\
		B^\prime_2 &= \{(v_2,w^\prime_1)\} \cup \{(z^\prime_i,w^\prime_{i+1}) \mid i = 1, \ldots, d-1\} \cup \{(z^\prime_d, r^\prime_1), (v_1,r^\prime_5)\}.
		\end{split}\label{addingB-odd'-2}
	\end{align}
	Similarly to Case 1, we have (\ref{addingB-odd':Ham}) and (\ref{addingB-odd':vertex}). 
	Furthermore, $B^\prime_1 \cup B^\prime_2$ consists of $(v_1, v_4)$- and $(v_2, v_5)$-paths if $d$ is even, and vertex-disjoint $(v_1, v_2)$- and $(v_4, v_5)$-paths if $d$ is odd. \medskip
	
	\textbf{Case 3}: In this case, we have $v_5=v_1$. Let $R^\prime_2 = \{(v_2,v_3), (v_3,v_4)\}$ (i.e., let $R_1$ be a $(v_2,v_4)$-path in $\mathcal{P}(T^\prime_1 \cap T^\prime_2)$), let $R^\prime_1$ denote $(v_1, r^\prime_1)$-path in $\mathcal{P}(T^\prime_1 \cap T^\prime_2)$. Then we have
	\begin{align*}
		\mathcal{P}(T^\prime_1)&= \{O^\prime_1, \dots, O^\prime_d\} \cup \{R^\prime_1\cup \{(v_1,v_2)\} \cup R^\prime_2\}\\
		\mathcal{P}(T^\prime_2)&= \{O^\prime_1, \dots, O^\prime_d\} \cup \{R^\prime_1\cup \{(v_1,v_4)\} \cup R^\prime_2\},
	\end{align*}
	where $R^\prime_1\cup \{(v_1,v_2)\} \cup R^\prime_2$ and $R^\prime_1\cup \{(v_1,v_4)\} \cup R^\prime_2$ are $(v_4,r^\prime_1)$- and $(v_2,r^\prime_1)$-paths, respectively. Define $B^\prime_1$ and $B^\prime_2$ by 
	\begin{align}
		\begin{split}
		B^\prime_1 &= \{(v_4,z^\prime_1)\} \cup \{(w^\prime_i,z^\prime_{i+1}) \mid i = 1, \ldots, d-1\} \cup \{(w^\prime_d, r^\prime_1)\} \\
		B^\prime_2 &= \{(v_2,w^\prime_1)\} \cup \{(z^\prime_i,w^\prime_{i+1}) \mid i = 1, \ldots, d-1\} \cup \{(z^\prime_d, r^\prime_1)\}.
		\end{split}\label{addingB-odd'-3}
	\end{align}
	Similarly to the previous cases, we have (\ref{addingB-odd':Ham}) and (\ref{addingB-odd':vertex}). 
	Furthermore, we have $B^\prime_1 \cup B^\prime_2$ is a $(v_2,v_4)$-path.
	
	In summary, we have the following lemma.
	\begin{lem}\label{lem:addingB-odd'}
		Let $B^\prime_1$ and $B^\prime_2$ be edge sets defined as above. Then they satisfy $(\ref{addingB-odd':Ham})$ and $(\ref{addingB-odd':vertex})$, and $B^\prime_1 \cup B^\prime_2$ consists of either $(\mathrm{i})$ vertex-disjoint $(v_1,v_4)$- and $(v_2,v_5)$-paths or $(\mathrm{ii})$ vertex-disjoint $(v_1,v_2)$- and $(v_4,v_5)$-paths if $|C^*|>4$, and $(\mathrm{iii})$ a $(v_2,v_4)$-path if $|C^*|=4$.
	\end{lem}
	\begin{figure}
	\centering
		\figexoddb
		\caption{Two edge sets $B_1$ and $B_2$ for path covers $T_1$ and $T_2$ in Fig. \ref{fig:exodd-st}.}\label{fig:exodd-b}
	\end{figure}
	\begin{figure}
	\centering
		\figexoddbp
		\caption{Two edge sets $B^\prime_1$ and $B^\prime_2$ for path covers $T^\prime_1$ and $T^\prime_2$ in Fig. \ref{fig:exodd-st'}.}\label{fig:exodd-b'}
	\end{figure}
	\textcolor{black}{Figures \ref{fig:exodd-b} and \ref{fig:exodd-b'} show an example of edge sets $B_1$, $B_2$, $B^\prime_1$, and $B^\prime_2$ for path covers $T_1$, $T_2$, $T^\prime_1$, and $T^\prime_2$ in Figs. \ref{fig:exodd-st} and \ref{fig:exodd-st'}.}
	Furthermore, $A^{(\prime)}_i$ and $B^{(\prime)}_i$ ($i=1,2$) satisfy the following properties.
	\begin{lem}\label{lem:addingAB-odd}
		Let $A_1$, $A_2$, $B_1$, and $B_2$ be defined as above. Then they are all pairwise disjoint, and $C = A_1 \cup A_2 \cup B_1 \cup B_2$ consists of either one or two cycles such that $V(C) = V\setminus \{v_2\}$. Furthermore, there exists a cycle $D$ such that $V(D) = V\setminus \{v_2\}$, $\ell(D) \geq \ell(C)$ and $(q, v_3) \in D$. 
	\end{lem}
	\begin{proof}
		It is not difficult to see that $A_1$, $A_2$, $B_1$, and $B_2$ are pairwise disjoint. 
Lemmas \ref{lem:4as2m-odd}, \ref{lem:addingA-odd}, and \ref{lem:addingB-odd} imply that $C=A_1 \cup A_2 \cup B_1 \cup B_2$ consists of either one or two cycles such that $V(C) = V\setminus \{v_2\}$. By $x_1=q$ and $(\ref{addingA-odd})$, we have $(q,v_3) \in C$. Thus if $C$ is a single cycle, the latter statement in the lemma holds. Assume that $C$ consists of two cycles. In this case, we can see that two edges $(v_1,z_1)$ and $(v_3,w_1)$ belong to different cycles by $(\ref{addingB-odd-1})$, $(\ref{addingB-odd-2})$, and $(\ref{addingB-odd-3})$. Let $D=(C \setminus \{(v_1,z_1), (v_3,w_1)\}) \cup \{(v_1,w_1), (v_3,z_1)\}$. Then $D$ is a cycle such that $V(D) = V\setminus \{v_2\}$. By assumption $(\ref{zw-ineq})$, we have $\ell(D)\geq\ell(C)$. Since $C$ contains $(q,v_3)$, so does $D$, which completes the proof.
	\end{proof}
	
	\begin{lem}\label{lem:addingAB-odd'}
		Let $A^\prime_1$, $A^\prime_2$, $B^\prime_1$, and $B^\prime_2$ be defined as above. Then they are all pairwise disjoint, and $C^\prime = A^\prime_1 \cup A^\prime_2 \cup B^\prime_1 \cup B^\prime_2$ consists of either one or two cycles such that $V(C^\prime) = V\setminus \{v_3\}$. Furthermore, there exists a cycle $D^\prime$ such that $V(D^\prime) = V\setminus \{v_3\}$, $\ell(D^\prime) \geq \ell(C^\prime)$ and $(q, v_2) \in D^\prime$. 
	\end{lem}
	\begin{proof}
		It is not difficult to see that $A^\prime_1$, $A^\prime_2$, $B^\prime_1$, and $B^\prime_2$ are pairwise disjoint. Lemmas \ref{lem:4as2m-odd'}, \ref{lem:addingA-odd'}, and \ref{lem:addingB-odd'} imply that $C^\prime=A^\prime_1 \cup A^\prime_2 \cup B^\prime_1 \cup B^\prime_2$ consists of either one or two cycles such that $V(C^\prime) = V\setminus \{v_3\}$. By $x^\prime_1 = q$ and $(\ref{addingA-odd'})$, we have $(q, v_2) \in C^\prime$. Thus if $C^\prime$ is a single cycle, the latter statement in the lemma holds. Assume that $C^\prime$ consists of two cycles. In this case, we can see that two edges $(v_4,z^\prime_1)$ and $(v_2,w^\prime_1)$ belong to different cycles by $(\ref{addingB-odd'-1})$, $(\ref{addingB-odd'-2})$, and $(\ref{addingB-odd'-3})$. Let $D^\prime=(C^\prime \setminus \{(v_4,z^\prime_1), (v_2,w^\prime_1)\}) \cup \{(v_4,w^\prime_1), (v_2,z^\prime_1)\}$. Then $D^\prime$ is a cycle such that $V(D^\prime) = V\setminus \{v_3\}$. By assumption (\ref{zw-ineq'}), we have $\ell(D^\prime)\geq\ell(C^\prime)$. Since $C^\prime$ contains $(q,v_2)$, so does $D^\prime$, which completes the proof.
	\end{proof}
	
	\begin{lem}\label{lem:2tours}
		Let $A^{(\prime)}_i$ and $B^{(\prime)}_i$ for $i=1,2$ be defined as above. Them there exist two tours $H$ and $H^\prime$ in $G$ such that $\ell(H) + \ell(H^\prime) \geq \ell(A_1)+ \ell(A_2) + \ell(B_1)+ \ell(B_2) + \ell(A^\prime_1)+ \ell(A^\prime_2) + \ell(B^\prime_1)+ \ell(B^\prime_2) + 2\ell(v_2,v_3)$.
	\end{lem}
	\begin{proof}
	Let $D$ and $D^\prime$ be a cycles in Lemmas \ref{lem:addingAB-odd} and \ref{lem:addingAB-odd'}, respectively. Then we have $V(D) = V \setminus \{v_2\}$, $V(D^\prime) = V \setminus \{v_3\}$, $(q, v_3) \in D$, and $(q, v_2) \in D^\prime$. Define $H$ and $H^\prime$ by
	\begin{align*}
		H &= (D \setminus \{(q,v_3)\}) \cup \{(q,v_2), (v_2,v_3)\}\\
		H^\prime &= (D^\prime \setminus \{(q,v_2)\}) \cup \{(q,v_3), (v_2,v_3)\}.
	\end{align*}
	Then $H$ and $H^\prime$ are tours. Furthermore, we have 
	\begin{align*}
		\ell(H) + \ell(H^\prime) &= \ell(D) + \ell(D^\prime) + 2\ell(v_2,v_3)\\
		&\geq \ell(A_1 \cup A_2 \cup B_1 \cup B_2) + \ell(A^\prime_1 \cup A^\prime_2 \cup B^\prime_1 \cup B^\prime_2) + 2\ell(v_2,v_3),
	\end{align*}
	which completes the proof.
	\end{proof}
	We are now ready to describe our approximation algorithm, called \texttt{TourOdd}. 	
	\todo[inline,caption={algorithm内のピリオド}]{algorithm内のピリオドに関して,"if" "for" "return" 等の予約語を含まない文,式にピリオドをつけるように統一しました．}
\begin{figure}[h!t]
		\begin{algorithm}[H]
			\caption{\texttt{TourOdd}}
			\label{alg:TourOdd}
			\begin{algorithmic}
			\Require A complete graph $G=(V,E)$ with odd $|V|$, and an edge length function $\ell :E \to \mathbb{R}_+$.
			\Ensure A tour $T_{\apx}$ in $G$.\vspace{5pt}
			\If{$n < 17$}
				\State Compute an optimal tour $T_{\opt}$ of $(G,\ell)$ by exhaustive search.
				\State Output $T_{\opt}$ and halt.
			\Else
				\State $\mathcal{T} := \emptyset$.
				\For{$v_1$,$v_2$,$v_3$ and $v_4$ in the 4-permutations of $V$}
				\State Compute a minimum weighted 2-factor $S$ among those containing $\{(v_1,v_2), (v_2,v_3), (v_3,v_4)\}$.
				\State Compute a minimum weighted path cover $T$ among those satisfying $(v_1, v_2), (v_2, v_3) \in T$ \State \ \ \ \ and $V_1(T)=V \setminus \{v_2\}$.
				\State Compute a minimum weighted path cover $T^\prime$ among those satisfying $(v_2, v_3), (v_3, v_4) \in T^\prime$ \State \ \ \ \ and $V_1(T^\prime)=V \setminus \{v_3\}$.
				\If{$S$ is a tour}
					\State $\mathcal{T} := \mathcal{T} \cup \{S\}$.
				\Else
					\State $S_1, T_1, S_2, T_2 := \texttt{FourPathCovers}(S,T)$.\vspace{1.5pt}
					\State Compute edge sets $A_1$, $A_2$, $B_1$, $B_2$ defined in (\ref{addingA-odd}), (\ref{addingB-odd-1}), (\ref{addingB-odd-2}), and (\ref{addingB-odd-3}).\vspace{1.5pt}
					\State $\mathcal{T} := \mathcal{T} \cup \{S_1 \cup A_1, S_2 \cup A_2, T_1 \cup B_1, T_2 \cup B_2 \}$.\vspace{1.5pt}
					\State $S^\prime_1, T^\prime_1, S^\prime_2, T^\prime_2 := \texttt{FourPathCovers}(S,T^\prime)$.\vspace{1.5pt}
					\State Compute edge sets $A^\prime_1$, $A^\prime_2$, $B^\prime_1$, $B^\prime_2$ defined in (\ref{addingA-odd'}), (\ref{addingB-odd'-1}), (\ref{addingB-odd'-2}), and (\ref{addingB-odd'-3}).\vspace{1.5pt}
					\State $\mathcal{T} := \mathcal{T} \cup \{S^\prime_1 \cup A^\prime_1, S^\prime_2 \cup A^\prime_2, T^\prime_1 \cup B^\prime_1, T^\prime_2 \cup B^\prime_2 \}$.\vspace{1.5pt}
				\EndIf
				\EndFor
				\State $T_{\apx} := \argmin_{T \in \mathcal{T}} \ell(T)$.\vspace{1.5pt}
				\State Output $T_{\apx}$ and halt.
			\EndIf
			\end{algorithmic}
		\end{algorithm}
	\end{figure}

	Before analyzation of $T_{\apx}$, let us evaluate $\ell(S)$, $\ell(T)$ and $\ell(T^\prime)$.
	\begin{lem}\label{lem:odd-opt}
		For a path $P = \{(v_1,v_2), (v_2,v_3), (v_3,v_4)\}$, let $S$, $T$ and $T^\prime$ be defined as above. If there exists an optimal tour that contains $P$, then
		\begin{align}
			2\ell(S) + \ell(T)+\ell(T^\prime) \leq 3\opt(G, \ell)+\ell(v_2,v_3).
		\end{align}
	\end{lem}
	\begin{proof}
		Obviously $\ell(S) \leq \opt(G,\ell)$ by definition. 
		Let $T_{\opt} = \{(v_i, v_{i+1}) \mid i=1, \ldots, n\}$ be an optimal tour of $(G, \ell)$, where $v_{n+1}=v_1$, and let $U$ and $U^\prime$ be two path covers defined by 
		\begin{align*}
			U&=\{(v_i,v_{i+1}) \mid i = 2,4,\ldots,n-1 \} \cup \{(v_1,v_2) \}\\
			U^\prime&=\{(v_{i},v_{i+1}) \mid i = 3,5,\ldots,n \} \cup \{(v_2,v_3) \}.
		\end{align*}
		Then by the definition of $T$ and $T^\prime$, $\ell(T) \leq \ell(U)$ and $\ell(T^\prime) \leq \ell(U^\prime)$. Therefore, we have
		\begin{align*}
			\ell(T)+\ell(T^\prime) \leq \ell(U)+\ell(U^\prime) = \opt(G, \ell) + \ell(v_2,v_3).
		\end{align*}
	\end{proof}
	\begin{thm}\label{thm:odd-delta}
	For a complete graph $G=(V,E)$ with an odd number of vertices and an edge length function $\ell:E \to \mathbb{R}_+$, \textbf{Algorithm} \texttt{TourOdd} computes a $3/4$-differential approximate tour of $(G,\ell)$ in polynomial time. 
	\end{thm}
	\begin{proof}
		If $n < 17$, \textbf{Algorithm} \texttt{TourOdd} clearly outputs an optimal tour in
constant time. Otherwise (i.e., $n\geq17$), let $T_{\opt}$ be an optimal tour of $(G, \ell)$ and let $P = \{(v_1, v_2), (v_2, v_3), (v_3, v_4)\}$ be a path contained in $T_{\opt}$. For this $P$, let $S$, $T$, and $T^\prime$ be defined as above. If $S$ is a tour, then $S$ is an optimal tour of $(G,\ell)$ and is output by the algorithm, which guarantees the statement of the theorem. On the other hand, if $S$ is not a tour, then we have
	\begin{align*}
		8\ell(T_{\apx}) &\leq \ell(S_1 \cup A_1)+\ell(S_2 \cup A_2)+\ell(T_1 \cup B_1)+\ell(T_2 \cup B_2)\\ &\quad+ \ell(S^\prime_1 \cup A^\prime_1)+\ell(S^\prime_2 \cup A^\prime_2)+\ell(T^\prime_1 \cup B^\prime_1)+\ell(T^\prime_2 \cup B^\prime_2)\\
		&= 2(2\ell(S) + \ell(T) + \ell(T^\prime)) + \ell(A_1 \cup A_2 \cup B_1 \cup B_2) + \ell(A^\prime_1 \cup A^\prime_2 \cup B^\prime_1 \cup B^\prime_2)\\
		&\leq 6\opt(G,\ell) + 2\wor(G,\ell),
	\end{align*}
	where the first equality follows from Lemmas \ref{lem:addingA-odd}, \ref{lem:addingA-odd'}, \ref{lem:addingB-odd}, and \ref{lem:addingB-odd'}, and the last inequality follows from Lemmas \ref{lem:2tours} and \ref{lem:odd-opt}.
	Thus $T_{\apx}$ is a 3/4-differential approximate tour of $(G,\ell)$. Note that $S$, $T$, and $T^\prime$ can be computed in polynomial time, since minimum weighted 1- and 2-factors can be computed in polynomial time. Furthermore, $A^{(\prime)}_i$ and $B^{(\prime)}_i$ for $i=1,2$ can be computed in polynomial time.
Thus \textbf{Algorithm} \texttt{TourOdd} is polynomial, which completes the proof.
	\end{proof}
	
\section*{Acknowledgement}
This work was partially supported by the joint project of Kyoto University and Toyota Motor Corporation, titled "Advanced Mathematical Science for Mobility Society" and by KAKENHI.

\bibliographystyle{plain}
\bibliography{A_3_4_Differential_Approximation_Algorithm_for_Traveling_Salesman_Problem}

\begin{thebibliography}{10}

\bibitem{arora1998polynomial}
Sanjeev Arora.
\newblock Polynomial time approximation schemes for euclidean traveling
  salesman and other geometric problems.
\newblock {\em Journal of the ACM}, 45(5):753--782, 1998.

\bibitem{ausiello2005completeness}
Giorgio Ausiello, Cristina Bazgan, Marc Demange, and Vangelis~Th Paschos.
\newblock Completeness in differential approximation classes.
\newblock {\em International Journal of Foundations of Computer Science},
  16(06):1267--1295, 2005.

\bibitem{bellman1961dynamic}
Richard Bellman.
\newblock Dynamic programming treatment of the travelling salesman problem.
\newblock {\em Journal of the ACM (JACM)}, 9(1):61--63, 1962.

\bibitem{bland1989large}
Robert~G Bland and David~F Shallcross.
\newblock Large travelling salesman problems arising from experiments in x-ray
  crystallography: a preliminary report on computation.
\newblock {\em Operations Research Letters}, 8(3):125--128, 1989.

\bibitem{chlebik2019approximation}
Miroslav Chleb{\'\i}k and Janka Chleb{\'\i}kov{\'a}.
\newblock Approximation hardness of travelling salesman via weighted
  amplifiers.
\newblock In {\em International Computing and Combinatorics Conference}, pages
  115--127. Springer, 2019.

\bibitem{christofides1976worst}
Nicos Christofides.
\newblock Worst-case analysis of a new heuristic for the travelling salesman
  problem.
\newblock Technical report, Carnegie-Mellon Univ Pittsburgh Pa Management
  Sciences Research Group, 1976.

\bibitem{cook2011pursuit}
William~J Cook.
\newblock {\em In pursuit of the traveling salesman: mathematics at the limits
  of computation}.
\newblock Princeton University Press, 2011.

\bibitem{demange1996approximation}
Marc Demange and Vangelis~Th Paschos.
\newblock On an approximation measure founded on the links between optimization
  and polynomial approximation theory.
\newblock {\em Theoretical Computer Science}, 158(1-2):117--141, 1996.

\bibitem{engebretsen2000clique}
Lars Engebretsen and Jonas Holmerin.
\newblock Clique is hard to approximate within $n^{1-O(1)}$.
\newblock In {\em International Colloquium on Automata, Languages, and
  Programming}, pages 2--12. Springer, 2000.

\bibitem{escoffier2008better}
Bruno Escoffier and J{\'e}r{\^o}me Monnot.
\newblock A better differential approximation ratio for symmetric {TSP}.
\newblock {\em Theoretical Computer Science}, 396(1-3):63--70, 2008.

\bibitem{grotschel1991optimal}
Martin Gr{\"o}tschel, Michael J{\"u}nger, and Gerhard Reinelt.
\newblock Optimal control of plotting and drilling machines: a case study.
\newblock {\em Zeitschrift f{\"u}r Operations Research}, 35(1):61--84, 1991.

\bibitem{hassin2001z}
Refael Hassin and Samir Khuller.
\newblock $z$-approximations.
\newblock {\em Journal of Algorithms}, 41(2):429--442, 2001.

\bibitem{held1962dynamic}
Michael Held and Richard~M Karp.
\newblock A dynamic programming approach to sequencing problems.
\newblock {\em Journal of the Society for Industrial and Applied mathematics},
  10(1):196--210, 1962.

\bibitem{lin1973effective}
Shen Lin and Brian~W Kernighan.
\newblock An effective heuristic algorithm for the traveling-salesman problem.
\newblock {\em Operations research}, 21(2):498--516, 1973.

\bibitem{little1963algorithm}
John~DC Little, Katta~G Murty, Dura~W Sweeney, and Caroline Karel.
\newblock An algorithm for the traveling salesman problem.
\newblock {\em Operations research}, 11(6):972--989, 1963.

\bibitem{monnot2002differential}
J{\'e}r{\^o}me Monnot.
\newblock Differential approximation results for the traveling salesman and
  related problems.
\newblock {\em Information Processing Letters}, 82(5):229--235, 2002.

\bibitem{monnot2003approximation}
J{\'e}r{\^o}me Monnot, Vangelis~Th Paschos, and Sophie Toulouse.
\newblock Approximation algorithms for the traveling salesman problem.
\newblock {\em Mathematical methods of operations research}, 56(3):387--405,
  2003.

\bibitem{monnot2003differential}
J{\'e}r{\^o}me Monnot, Vangelis~Th Paschos, and Sophie Toulouse.
\newblock Differential approximation results for the traveling salesman problem
  with distances 1 and 2.
\newblock {\em European Journal of Operational Research}, 145(3):557--568,
  2003.

\bibitem{monnot2014traveling}
J{\'e}r{\^o}me Monnot and Sophie Toulouse.
\newblock The traveling salesman problem and its variations.
\newblock {\em Paradigms of Combinatorial Optimization: Problems and New
  Approaches}, pages 173--214, 2014.

\bibitem{papadimitriou1998combinatorial}
Christos~H Papadimitriou and Kenneth Steiglitz.
\newblock {\em Combinatorial optimization: algorithms and complexity}.
\newblock Courier Corporation, 1998.

\bibitem{punnen2007traveling}
Abraham~P Punnen.
\newblock The traveling salesman problem: Applications, formulations and
  variations.
\newblock In {\em The traveling salesman problem and its variations}, pages
  1--28. Springer, 2007.

\bibitem{shmoys1985traveling}
DB~Shmoys, JK~Lenstra, AHG~Rinnooy Kan, and EL~Lawler.
\newblock {\em The traveling salesman problem}, volume~12.
\newblock Wiley, 1985.

\end{thebibliography}
\newpage
\listoftodos[Notes]
\end{document}